\documentclass[5p,times,number,sort&compress]{elsarticle}

\usepackage{amsmath,amssymb,amsthm,mathtools,bm}
\usepackage{booktabs}
\usepackage{enumitem}
\usepackage[ruled]{algorithm}
\usepackage{algpseudocode}
\usepackage{graphicx}
\usepackage{array}
\usepackage{xcolor}
\usepackage[final]{microtype}
\journal{Automatica}
\graphicspath{{revision_artifacts/}}
\setlist[enumerate]{topsep=2pt,itemsep=1pt,parsep=0pt,partopsep=0pt}

\newtheorem{assumption}{Assumption}

\newtheorem{lemma}{Lemma}
\newtheorem{theorem}{Theorem}
\newtheorem{corollary}{Corollary}

\newtheorem{proposition}{Proposition}

\newcommand{\SO}{\mathrm{SO}(3)}
\newcommand{\R}{\mathbb{R}}
\newcommand{\Exp}{\mathrm{Exp}}
\newcommand{\Log}{\mathrm{Log}}
\newcommand{\calX}{\mathcal{X}}
\newcommand{\calU}{\mathcal{U}}
\newcommand{\calN}{\mathcal{N}}
\newcommand{\calM}{\mathcal{M}}
\newcommand{\calG}{\mathcal{G}}
\newcommand{\calV}{\mathcal{V}}
\newcommand{\calE}{\mathcal{E}}
\newcommand{\calD}{\mathcal{D}}
\newcommand{\col}{\mathrm{col}}

\newcommand{\data}{\mathrm{data}}
\newcommand{\capb}{\mathrm{cap}}
\definecolor{revisionblue}{RGB}{0,0,0}
\newenvironment{revision}{\begingroup\color{revisionblue}}{\endgroup}
\newenvironment{retained}{\begingroup\color{black}}{\endgroup}

\begin{document}

\begin{frontmatter}

\title{Certificate-Carrying Distributed Model Predictive Control on Product Manifolds with $\SO$\tnoteref{t1}}
\tnotetext[t1]{This work was partially supported by the National Natural Science
Foundation of China under Grant 62133003 and the National Social Science Fund of China
under Grant 24CTJ010.}

\author[hubu]{Shengjun Zhang}
\ead{sj.zhang@hubu.edu.cn}
\author[whu]{Tingyi Liu}
\ead{2019101050047@whu.edu.cn}
\author[kth]{Lei Xu}
\ead{lei5@kth.se}
\author[neu]{Tao Yang\corref{cor1}}
\ead{yangtao@mail.neu.edu.cn}
\cortext[cor1]{Corresponding author.}

\address[hubu]{School of Artificial Intelligence, Hubei University, Wuhan 430061, China}
\address[whu]{School of Economics and Management, Wuhan University, Wuhan 430072, China}
\address[kth]{Division of Decision and Control Systems, KTH Royal Institute of
Technology, and Digital Futures, Stockholm, Sweden}
\address[neu]{State Key Laboratory of Synthetical Automation for Process Industries,
Northeastern University, Shenyang 110819, China}

\begin{abstract}
\color{revisionblue}
This paper studies constraint certification in synchronous distributed model predictive
control (DMPC) when neighboring predictions change between sampling instants. Before
the parallel local solves, each agent communicates a shifted prediction and an announced
update budget. A hard trajectory trust region makes that budget enforceable, while an
edge-wise feasibility cap computed from the shifted packets keeps the fallback feasible
without using any current optimizer output. Distance and relative-attitude constraints
are tightened with explicit Lipschitz constants and two budget layers: one accounts for
the simultaneous neighbor update and the other retains a checkable shift reserve. We
prove hard pairwise constraint satisfaction and recursive feasibility under stated
nominal-execution and terminal assumptions, give the additional residual caused by
execution error, and derive a local practical value-decrease bound. A spacecraft
formation example uses hard terminal and pairwise constraints, a geodesic relative-
attitude constraint on $\SO$, and reproducible terminal-set checks. Comparisons with
fixed, trajectory-only, and windowed online margins show that the proposed budget
reduces conservatism while preserving a positive shifted-feasibility margin.
\end{abstract}


\end{frontmatter}

\section{Introduction}

Distributed model predictive control (DMPC) coordinates networked systems by solving
local finite-horizon problems that depend on communicated neighbor predictions
\cite{Dunbar2007,Keviczky2006,Scattolini2009,Christofides2013,FarinaScattolini2012}.
When all agents update in parallel, however, the prediction used as a parameter by one
agent generally differs from the trajectory eventually accepted by its neighbor. A
pairwise constraint that is feasible against the communicated parameter need not remain
feasible after both agents update.

As in centralized stabilizing MPC, recursive feasibility also depends on a shifted input
sequence, a terminal invariant set, and a terminal decrease condition
\cite{ChenAllgower1998,Mayne2000,Rawlings2017,GrunePannek2017}. In a distributed problem,
however, the shifted own sequence is evaluated against neighbor parameters that may be
changed by simultaneous optimizations. Terminal ingredients alone do not resolve that
additional mismatch. The construction below therefore separates the familiar local
terminal argument from a new edge-reserve argument for the communicated trajectories.

\begin{revision}
Fixed robust tightening prescribes an offline uncertainty radius
\cite{RichardsHow2007,MayneSeronRakovic2005}. Related adaptive constructions address
different objects. Giselsson--Rantzer tighten a convex dual-decomposition problem as a
function of finite inner iterations; K\"ohler--M\"uller--Allg\"ower assume bounded
suboptimality and constraint violation from inexact dual optimization
\cite{GiselssonRantzer2014,KohlerMullerAllgower2019}. The Trodden formulations concern
linear systems and propagate tube or mutual-disturbance sets, in the latter case
exchanging constraint sets rather than planned trajectories
\cite{Trodden2014,TroddenMaestre2017}.

Here $\rho_i$ instead bounds the coupling-output displacement of an accepted nonlinear
trajectory from a shift already stored by neighbors. It is available before the solve,
enforced as a hard trajectory constraint, and verified afterward. Directly transplanting
the cited algorithms would change the plant class, distributed optimization loop,
communicated object, and source of the feasibility margin simultaneously. The five-sample
windowed baseline in Section~5 therefore keeps the same nonlinear $\SO$ problem and the
same measured displacement, changing only the online update rule. This is the closest
controlled comparison for the mechanism studied here, not a claimed ranking over tube
or inexact-dual DMPC.

The timing is causal only if the current optimizer is absent from every current problem
parameter. Each sample therefore uses two pre-solve communication rounds and one
post-solve round. Agents first exchange committed shifts, then compute and announce the
admissible budgets used by the parallel NLPs. After the solves, they exchange the
accepted trajectories and measured updates used at the next sample. Neighbors store the
two trajectories and recompute $\mu_i(t)$ rather than trusting a reported scalar.
\end{revision}

Rigid-body networks motivate the setting because their states naturally contain
$\SO$ factors. Geometric MPC avoids singular attitude coordinates and preserves the
group constraint \cite{BulloLewis,Lee2010,Kalabic2017}. The budget mechanism itself is
metric rather than specific to a Lie group. In this paper, the geometry enters the
dynamics, the coupling metric, the terminal coordinates, and a genuinely manifold-
dependent relative-attitude constraint. All stability statements are local, consistently
with the topology of $\SO$ \cite{BhatBernstein2000}.

\begin{revision}
The contributions are as follows.
\begin{enumerate}[label=(\roman*)]
\item We give a synchronous, noncircular DMPC protocol with separately timed pre- and
post-solve packets. A scalar feasibility cap, computed from the committed shifted
packets before optimization, converts the online reserve test into an explicit closed-
loop mechanism.
\item We derive edge-wise hard-constraint and recursive-feasibility guarantees for the
resulting two-layer tightening. Exact nominal execution is stated explicitly; for
model, estimation, or sampling error we give the additional residual and the margin
needed to recover a hard guarantee. The terminal, trust, and safety constraints in the
certified formulation are hard, so the proof does not rely on unbounded or penalized
slacks.
\item We instantiate the method for spacecraft position-attitude dynamics with an
intrinsic relative-attitude constraint. The numerical study reports explicit
Lipschitz constants, a windowed online-tightening baseline, a boundary challenge with a
forced solver rejection, five fixed-seed runs, a six-agent graph, an explicit numerical
dissipation decomposition, and sampled tests of the projected terminal law.
\end{enumerate}

We do not claim that a Lipschitz radius gives a globally smallest physical margin.
Directional, stage-dependent, gradient, or set-valued information can be less
conservative. The contribution is the enforceable and verifiable use of a
trajectory-specific scalar budget, together with the feasibility-cap mechanism.
\end{revision}

\section{Setting and Explicit Margins}

\subsection{Notation and local product coordinates}

The special orthogonal group and its tangent algebra are
\begin{align}
\SO&=\{R\in\R^{3\times3}:R^\top R=I,\ \det R=1\},\nonumber\\
\mathfrak{so}(3)&=\{A\in\R^{3\times3}:A^\top=-A\}.
\label{eq:so3_definition}
\end{align}
For $a\in\R^3$, the hat map satisfies $a^\wedge b=a\times b$, and
$({\cdot})^\vee$ denotes its inverse. We write
$\Exp(a)=\exp(a^\wedge)$ and use $\Log(R)^\vee$ for the principal local logarithm.
All logarithms below are evaluated on a geodesic ball of radius strictly smaller than
$\pi$, where this coordinate is single valued. The Euclidean norm is $\|\cdot\|$;
vector inequalities are componentwise.

A representative agent state is
\begin{equation}
x_i=(p_i,\xi_i,R_i,\omega_i)
\in\R^{n_p}\times\R^{n_\xi}\times\SO\times\R^3,
\label{eq:product_state}
\end{equation}
where $p_i$ contains the position variables that enter pairwise constraints, $\xi_i$
collects remaining Euclidean states, and $\omega_i$ is expressed in the body frame.
Multiple rotation factors are covered by taking products of the corresponding local
coordinates. Given a desired equilibrium
$x_{i,d}=(p_{i,d},\xi_{i,d},R_{i,d},\omega_{i,d})$, define
\begin{equation}
e_i(x_i)=\col\!\left(p_i-p_{i,d},\ \xi_i-\xi_{i,d},\
\Log(R_{i,d}^\top R_i)^\vee,\ \omega_i-\omega_{i,d}\right).
\label{eq:local_error}
\end{equation}
These coordinates are used only for local costs and terminal ingredients; rotations are
propagated intrinsically and are never optimized as unconstrained attitude vectors.

Two distances are needed below. Cost perturbations and terminal propagation use the
full-state trajectory metric
\begin{align}
d_i(x_i,y_i)&=\|p_i-q_i\|+\|\xi_i-\zeta_i\|+
d_R(R_i,S_i)+\|\omega_i-\varpi_i\|,\nonumber\\
d_{H,i}(X_i,Y_i)&=\max_{0\leq k\leq H}d_i(x_i(k),y_i(k)).
\label{eq:full_metric}
\end{align}
Here $y_i=(q_i,\zeta_i,S_i,\varpi_i)$. Unit component weights are used in the
reported full-state diagnostic; any fixed positive weights give an equivalent local
metric and only rescale the regularity constants. The constraint certificate uses the
smaller coupling-output distance below, so velocity changes that do not enter an edge
constraint do not enlarge its tightening.

\subsection{Product-manifold dynamics}

The communication graph is $\calG=(\calV,\calE)$, where
$\calV=\{1,\ldots,N\}$, $\calN_i$ denotes the neighbors of agent $i$, and
$\calN_i^+=\calN_i\cup\{i\}$. Agent $i$
has state $x_i\in\calM_i$, input $u_i\in\calU_i$, and sampled dynamics
\begin{equation}
x_i^+=f_i(x_i,u_i), \qquad
\calM_i=\R^{n_i}\times\SO^{m_i}.
\label{eq:dynamics}
\end{equation}
The map $f_i$ is assumed continuous and locally smooth on the operating set. The
network state belongs to $\calM=\prod_i\calM_i$.

For $R,S\in\SO$, define the geodesic distance
\begin{equation}
d_R(R,S)=\|\Log(R^\top S)^\vee\|,
\label{eq:rotation_metric}
\end{equation}
on a ball with radius strictly below $\pi$. Here $a^\wedge b=a\times b$,
$\Exp(a):=\exp(a^\wedge)$, and $\Log(R)^\vee$ is the inverse local coordinate. For a
position-attitude coupling output, define
\begin{equation}
d_i^c(x_i,y_i)=
\max\{\|p_i-q_i\|,\,\ell_R d_R(R_i,S_i)\},
\label{eq:coupling_metric}
\end{equation}
where $\ell_R>0$ converts radians to the position unit. Velocity components can be
added if a coupling constraint uses them. Equation~\eqref{eq:coupling_metric} is a
metric on the coupling-output space and a pseudometric on the full state. For
trajectories of length $H+1$,
\begin{equation}
d_{H,i}^c(X_i,Y_i)=\max_{0\leq k\leq H}d_i^c(x_i(k),y_i(k)).
\label{eq:trajectory_metric}
\end{equation}

For every directed edge $(i,j)$ and constraint index $q$, hard coupling constraints
have the form
\begin{equation}
g_{ij}^q(x_i,x_j)\leq 0.
\label{eq:coupling_constraint}
\end{equation}
On a compact local set $\calD$, assume
\begin{equation}
|g_{ij}^q(x_i,y_j)-g_{ij}^q(x_i,z_j)|
\leq L_{ij}^q d_j^c(y_j,z_j).
\label{eq:lipschitz_constraint}
\end{equation}
Only the neighbor argument is perturbed in \eqref{eq:lipschitz_constraint}; the own
prediction remains a decision variable in the local problem.

\textcolor{revisionblue}{Packets are exchanged in both directions on each constrained
physical edge.} The directed notation records constraint ownership: agent $i$ optimizes
its state against the committed packet of $j$, so $L_{ij}^q$ measures sensitivity only
to the second argument.
If the same physical edge is represented in both local problems, the two directed
constraints and their constants are checked separately. No convexity of $g_{ij}^q$ is
required for the certificate implication; local Lipschitz continuity on the declared
compact set is sufficient.

\subsection{Constants used for distance and attitude constraints}

\begin{revision}
The simulations use distance, rather than squared-distance, constraints:
\begin{align}
g_{ij}^{\rm col}(p_i,p_j)&=d_{\rm safe}-\|p_i-p_j\|,\label{eq:gcol}\\
g_{ij}^{\rm com}(p_i,p_j)&=\|p_i-p_j\|-d_{\rm comm}.
\label{eq:gcom}
\end{align}
To expose the constant rather than merely state it, put
$a=p_i-p_j$ and $b=p_i-q_j$. The triangle inequality applied in both directions gives
\begin{align}
\|a\|&\leq\|b\|+\|p_j-q_j\|,\nonumber\\
\|b\|&\leq\|a\|+\|p_j-q_j\|,\nonumber\\
\big|\|p_i-p_j\|-\|p_i-q_j\|\big|
&\leq\|p_j-q_j\|.
\label{eq:distance_reverse_triangle}
\end{align}
The additive constants and opposite signs in \eqref{eq:gcol}--\eqref{eq:gcom} disappear
after taking absolute values. Since
$\|p_j-q_j\|\leq d_j^c(x_j,y_j)$, \eqref{eq:distance_reverse_triangle} proves
$L_{ij}^{\rm col}=L_{ij}^{\rm com}=1$. This removes the invalid small constant that
would result from scaling a fixed margin for a squared-distance constraint.

Let $R_{ij}^{d}$ be the desired relative attitude and define
\begin{equation}
g_{ij}^{R}(R_i,R_j)
=d_R(R_i^\top R_j,R_{ij}^{d})-\theta_{\max}.
\label{eq:gatt}
\end{equation}
Bi-invariance and the reverse triangle inequality imply
\begin{align}
&|g_{ij}^{R}(R_i,R_j)-g_{ij}^{R}(R_i,S_j)|\nonumber\\
&\quad\leq d_R(R_i^\top R_j,R_i^\top S_j)
=d_R(R_j,S_j)\nonumber\\
&\quad\leq\ell_R^{-1}d_j^c(x_j,y_j).
\label{eq:attitude_lipschitz_derivation}
\end{align}
hence $L_{ij}^{R}=1/\ell_R$. These constants are analytic and do not depend on a
tuned robust margin.

For later reference, if both arguments change, two applications of the same inequality
give
\begin{equation}
g_{ij}^q(y_i,y_j)\leq g_{ij}^q(x_i,x_j)
+L_{ij}^{q,i}d_i^c(y_i,x_i)+L_{ij}^q d_j^c(y_j,x_j),
\label{eq:two_argument_lipschitz}
\end{equation}
where $L_{ij}^{q,i}$ is the first-argument constant. The nominal synchronous theorem
needs only $L_{ij}^q$ because each local problem keeps its own trajectory as a decision
variable. Equation~\eqref{eq:two_argument_lipschitz} is used when actual execution errors
of both agents are included.
\end{revision}

\section{Prediction-Budget DMPC}

\subsection{Shifted, optimal, and accepted predictions}

At time $t$, $\bar X_i(t)=\{\bar x_i(k|t)\}_{k=0}^{H}$ denotes the committed
shifted trajectory stored by agent $i$ and its neighbors. If $\widehat X_i(t-1)$ was
accepted at the preceding step, then
\begin{align}
\bar x_i(k|t)&=\widehat x_i(k+1|t-1), &&0\leq k<H,\label{eq:shift1}\\
\bar x_i(H|t)&=f_i(\widehat x_i(H|t-1),
\kappa_i(\widehat x_i(H|t-1))).
\label{eq:shift2}
\end{align}
Under exact nominal execution, $\bar x_i(0|t)=x_i(t)$. The local optimizer is denoted
$X_i^\star(t)$. The accepted trajectory $\widehat X_i(t)$ equals the optimizer only if
all hard numerical checks pass; otherwise it equals the feasible shifted fallback.

\begin{revision}
At $t=0$, the committed trajectories are initialized by a jointly admissible warm start
whose pairwise reserves are nonnegative. This is the usual initial-feasibility requirement
in stabilizing MPC, but here it is directly testable before any distributed solve. At
later instants, the induction in Theorem~\ref{thm:recursive} constructs the warm start.
Three trajectories are therefore kept distinct throughout: $\bar X_i$ is known before
the solve, $X_i^\star$ is returned by the numerical optimizer, and $\widehat X_i$ is the
candidate actually accepted after independent hard-constraint checks. Conflating these
objects would recreate the same-step circularity that the packet timing is intended to
remove.
Table~\ref{tab:notation} summarizes the timing: its first three entries are known or
formed before the local solve, whereas the last two result from the solve and the
subsequent acceptance checks.
\end{revision}

\begin{table*}[t]
\color{revisionblue}
\centering
\caption{Timing and roles of prediction and packet quantities.}
\label{tab:notation}
\begin{tabular}{lll}
\toprule
Symbol & Meaning & Stage in sample $t$ \\
\midrule
$\bar X_i(t)$ & committed shifted prediction in \eqref{eq:shift1}--\eqref{eq:shift2}
 & known before local OCP \\
$\rho_i^{\data}(t)$ & update budget requested from past measured updates
 & computed before local OCP \\
$\rho_i(t)$ & announced budget after the feasibility cap
 & fixed before local OCP \\
$X_i^\star(t)$ & local optimizer & returned by local OCP solver \\
$\widehat X_i(t),\mu_i(t),\chi_i(t)$ & accepted trajectory, coupling update, and
full-state update
 & formed after acceptance checks \\
\bottomrule
\end{tabular}
\end{table*}

\subsection{Data request and feasibility cap}

\begin{revision}
For $t\geq1$, the previous measured update affects only the next request. One simple
rule is
\begin{equation}
\rho_i^{\data}(t)=\min\{\rho_i(t-1),
\max\{\rho_{\min},\mu_i(t-1)+\delta_\rho\}\},
\label{eq:data_budget}
\end{equation}
where $\delta_\rho\geq0$. A windowed baseline replaces $\mu_i(t-1)$ by the maximum
of the last $w$ measured updates.

Before optimization, all agents first exchange $\bar X_i(t)$. Using the stored
neighbor shifts and the public edge functions and constants, each agent then evaluates
the incident edge-stage caps needed for its own budget. For every directed edge, stage,
and constraint, define the available shifted-packet reserve
\begin{equation}
r_{ij,k}^q(t)=
-g_{ij}^q(\bar x_i(k|t),\bar x_j(k|t))-\epsilon_{ij}^q,
\label{eq:raw_reserve}
\end{equation}
where $\epsilon_{ij}^q>0$ covers the numerical feasibility tolerance. For
$L_{ij}^q>0$, using the committed shift itself in the two-layer constraint is
equivalent to
\begin{align}
&g_{ij}^q(\bar x_i(k|t),\bar x_j(k|t))
+2L_{ij}^q\rho_j(t)+\epsilon_{ij}^q\leq0\nonumber\\
&\quad\Longleftrightarrow\quad
0\leq\rho_j(t)\leq\frac{r_{ij,k}^q(t)}{2L_{ij}^q}.
\label{eq:cap_derivation}
\end{align}
Thus each edge-stage inequality supplies an explicit admissible interval for the same
neighbor budget. Intersecting all of these intervals gives the scalar cap
\begin{equation}
\rho_j^{\capb}(t)=
\min_{\substack{i:j\in\calN_i,\,1\leq k\leq H\\q:L_{ij}^q>0}}
\frac{r_{ij,k}^q(t)}{2L_{ij}^q}.
\label{eq:budget_cap}
\end{equation}
The convention is $\min\varnothing=+\infty$ when agent $j$ enters no constraint with a
positive neighbor-argument constant.
If any reserve in \eqref{eq:raw_reserve} is negative, the pre-solve feasibility test
fails. Under the assumptions of Theorem~\ref{thm:recursive}, this does not occur.
The announced budget is
\begin{equation}
\rho_j(t)=\max\{0,\min\{\rho_j^{\data}(t),
\rho_j^{\capb}(t)\}\}.
\label{eq:announced_budget}
\end{equation}
A positive floor is a performance preference and is applied only when permitted by the
cap. Thus feasibility takes precedence if the available reserve becomes small.

The logical pre-solve certificate record is
\begin{equation}
\mathcal{C}_i^-(t)=(t,\bar X_i(t),\rho_i(t)).
\label{eq:prepacket}
\end{equation}
It is assembled after the two pre-solve rounds and contains no quantity generated by
the time-$t$ optimizer. The first round transmits $(t,\bar X_i(t))$; because that
trajectory is already stored, the second round transmits only $(t,\rho_i(t))$ rather
than sending $\bar X_i(t)$ again.
\end{revision}

\subsection{Hard local problem}

Given the frozen pre-solve packets, all local problems are solved in parallel. Agent $i$
uses the local error \eqref{eq:local_error}. Representative relative errors evaluated
against the committed packet of $j$ are
\begin{align}
e_{ij}^{p}(k)&=(\bar p_j(k)-p_i(k))-(p_{j,d}-p_{i,d}),\nonumber\\
e_{ij}^{R}(k)&=\Log\!\left((R_{ij}^{d})^\top
R_i(k)^\top\bar R_j(k)\right)^\vee .
\label{eq:relative_errors}
\end{align}
The stage cost may then be written as
\begin{align}
\ell_i(k)={}&\|e_i(x_i(k))\|_{Q_i}^2+
\|u_i(k)-u_{i,d}\|_{R_i^u}^2 \nonumber\\
&+\sum_{j\in\calN_i}\left(\|e_{ij}^{p}(k)\|_{Q_{ij}^{p}}^2+
\|e_{ij}^{R}(k)\|_{Q_{ij}^{R}}^2\right)\nonumber\\
&+\|u_i(k)-u_i(k-1)\|_{S_i}^2+
\|e_i^{\rm sh}(k)\|_{W_i}^2,
\label{eq:stage_cost}
\end{align}
where $e_i^{\rm sh}$ contains the position and logarithmic-attitude differences from
$\bar x_i(k|t)$. Terms not needed by a particular application can be omitted. Positive
definite self-state and input weights give the local stage-cost lower bound used later;
the relative and shift terms improve coordination but are not substituted for that
condition. In particular, we write
\[
\ell_i^0(x_i,u_i):=\|e_i(x_i)\|_{Q_i}^2+
\|u_i-u_{i,d}\|_{R_i^u}^2
\]
for the packet-independent positive-definite part of the stage cost. At $k=0$,
$u_i(-1)$ denotes the input applied at the preceding sampling instant.

For terminal ingredients, let
\begin{equation}
V_{f,i}(x_i)=e_i(x_i)^\top P_i e_i(x_i),\qquad
\calX_{f,i}=\{x_i:V_{f,i}(x_i)\leq c_i\},
\label{eq:terminal_set}
\end{equation}
and let $\kappa_i$ be a hard-admissible local terminal controller. In the numerical
study it is the equilibrium input plus an LQR correction, projected componentwise onto
$\calU_i$. This is the standard local terminal-set construction used in stabilizing MPC
\cite{Mayne2000,Rawlings2017}, with the invariance and decrease conditions stated
explicitly in Assumption~\ref{ass:terminal}. The local finite-horizon objective is
\begin{equation}
J_i=\sum_{k=0}^{H-1}\ell_i(k)
+V_{f,i}(x_i(H))
\label{eq:cost}
\end{equation}
\begin{revision}
and the complete certified problem is
\begin{align}
\mathcal P_i(t):\quad\min_{U_i}\quad&J_i \nonumber\\
\mathrm{s.t.}\quad
x_i(0|t)&=x_i(t),\nonumber\\
x_i(k+1|t)&=f_i(x_i(k|t),u_i(k|t)),\nonumber\\
x_i(k|t)&\in\calX_i\cap\calD_i,\quad u_i(k|t)\in\calU_i,
\label{eq:local_problem}
\end{align}
together with the following hard trust, coupling, and terminal constraints:
\begin{align}
d_{H,i}^c(X_i,\bar X_i(t))&\leq\rho_i(t),
\label{eq:trust}\\
g_{ij}^q(x_i(k),\bar x_j(k|t))&\leq-\eta_{ij}^q(t),
\label{eq:tightened}\\
\eta_{ij}^q(t)&=2L_{ij}^q\rho_j(t)+\epsilon_{ij}^q,
\label{eq:two_layer}\\
V_{f,i}(x_i(H))&\leq c_i.
\label{eq:hard_terminal}
\end{align}
Constraint \eqref{eq:tightened} is imposed for $j\in\calN_i$, all $q$, and
$k=1,\ldots,H$; the local state and input conditions in
\eqref{eq:local_problem} use their natural stage ranges. The first layer
$L_{ij}^q\rho_j$ in \eqref{eq:two_layer} covers the simultaneous update of neighbor
$j$. The second, identical layer remains as reserve in the accepted trajectory and
allows the next shifted packet to be tested by the pre-solve cap. The small $\epsilon_{ij}^q$ is a
declared numerical tolerance, not a replacement for either layer.

Pairwise, trust, and terminal slacks are fixed to zero in the certified mode. If an
application instead enables a soft pairwise slack $s_{ij}^q$, the corresponding
conclusion is only $g_{ij}^q\leq s_{ij}^q$ and lies outside the hard-safety theorem.
This distinction is operational: a feasible soft NLP is not reported as a certified
hard-feasible solve. Likewise, an optimizer termination flag alone is insufficient;
the candidate is rolled out again and every displayed hard inequality is evaluated to
the declared acceptance tolerance.
\end{revision}

\begin{revision}
After optimization, let $\mathcal A_i(t)$ be the finite set containing the committed
hard-feasible fallback $U_i^{\rm fb}(t)$ and every independently rolled-out optimizer
candidate that passes all hard checks. The accepted candidate and its realized feasible
cost are
\begin{align}
\widehat U_i(t)&\in\arg\min_{U_i\in\mathcal A_i(t)}
J_i(U_i;x_i(t),\bar X_{\calN_i^+}(t)),
\label{eq:acceptance_rule}\\
\widehat J_i(t)&:=J_i(\widehat U_i(t);x_i(t),
\bar X_{\calN_i^+}(t))
\leq J_i(U_i^{\rm fb}(t);x_i(t),\bar X_{\calN_i^+}(t)).
\label{eq:accepted_value}
\end{align}
This finite comparison does not assume that the numerical NLP output is a global
minimizer. Agent $i$ then records the coupling-output and full-state displacements
\begin{equation}
\mu_i(t)=d_{H,i}^c(\widehat X_i(t),\bar X_i(t))\leq\rho_i(t),
\qquad
\chi_i(t)=d_{H,i}(\widehat X_i(t),\bar X_i(t)).
\label{eq:measured_update}
\end{equation}
and sends the post-solve packet
\begin{equation}
\mathcal{C}_i^+(t)=(t,\widehat X_i(t),\mu_i(t)).
\label{eq:postpacket}
\end{equation}
A neighbor independently verifies both distances in \eqref{eq:measured_update} from the
newly received trajectory and the stored $\bar X_i(t)$. Only $\mu_i$ enters the
constraint budget; $\chi_i$ is used in the value-perturbation analysis and need not be
sent as an additional scalar. The nominal protocol assumes that this synchronized
exchange succeeds. A failed or missing packet is rejected; the stored shift alone is
not treated as a certificate of the sender's realized update.

The two packets have complementary roles. The pre-solve packet is a commitment: its
budget is already feasible and is enforced in \eqref{eq:trust}. The post-solve packet is
a verification record: it cannot change the problem just solved, but it allows every recipient
to recompute the realized update and form the next request. A packet is accepted only if
its time stamp, trajectory length, manifold consistency, and inequality
\eqref{eq:measured_update} all pass. The scalar itself is therefore not trusted without
the trajectory data needed to verify it.
Algorithm~\ref{alg:protocol} displays the three communication barriers explicitly.
\end{revision}

\begin{algorithm}[t]
\color{revisionblue}
\caption{Synchronous prediction-budget DMPC.}
\label{alg:protocol}
\color{revisionblue}\footnotesize
\begin{algorithmic}[1]
\Require $x_i(t)$, $\widehat X_i(t-1)$, and $\mu_i(t-1)$ for all $i\in\calV$
\Ensure $u_i(t)$ and verified packets $\mathcal C_i^-(t),\mathcal C_i^+(t)$
\ForAll{$i\in\calV$ \textbf{in parallel}}
  \State $\bar X_i(t)\gets S_i\widehat X_i(t-1)$; broadcast $\bar X_i(t)$
\EndFor
\State \textbf{synchronize}; store all neighboring shifts
\ForAll{$j\in\calV$ \textbf{in parallel}}
  \State evaluate incident caps from \eqref{eq:raw_reserve}; set $\rho_j^{\capb}(t)$ to their minimum
  \State project $\rho_j(t)$ by \eqref{eq:announced_budget}; broadcast $(t,\rho_j(t))$
\EndFor
\State \textbf{synchronize}; assemble and freeze all $\mathcal C_j^-(t)$
\ForAll{$i\in\calV$ \textbf{in parallel}}
  \State solve $\mathcal P_i(t)$, form $\mathcal A_i(t)$, and select $\widehat U_i(t)\in\arg\min_{U_i\in\mathcal A_i(t)}J_i(U_i)$
  \State compute $\mu_i(t),\chi_i(t)$ by \eqref{eq:measured_update}; broadcast $\mathcal C_i^+(t)$
\EndFor
\State \textbf{synchronize}; verify each $\mathcal C_i^+(t)$
\ForAll{$i\in\calV$ \textbf{in parallel}}
  \State apply $u_i(t)=\widehat u_i(0|t)$ and store $S_i\widehat X_i(t)$
\EndFor
\end{algorithmic}
\end{algorithm}

\section{Main Results}
\begingroup\small

\begin{revision}
\subsection{Assumptions}

\begin{assumption}
For each agent there is a compact, locally geodesically convex set
$\calD_i\subset\calM_i$, contained in an $\SO$ chart of radius strictly below $\pi$.
On $\calD=\prod_i\calD_i$, the dynamics and terminal controllers are locally Lipschitz,
the packet-dependent costs are continuously differentiable, and
\eqref{eq:lipschitz_constraint} holds. Every local problem imposes
$x_i(k|t)\in\calD_i$ as a hard prediction constraint, and the initialized MPC problem
is feasible.
\label{ass:regularity}
\end{assumption}

\begin{assumption}
During the certified nominal operation, the applied input, sampled model, state estimate,
and packet timing are exact, so
$x_i(t+1)=\widehat x_i(1|t)$. The nonideal case is treated explicitly in
Corollary~\ref{cor:nonideal}.
\label{ass:execution}
\end{assumption}

\begin{assumption}
For $\calX_{f,i}=\{x_i:V_{f,i}(x_i)\leq c_i\}\subset\calX_i\cap\calD_i$, the projected terminal controller
$\kappa_i$ satisfies the hard input and state constraints,
$f_i(x_i,\kappa_i(x_i))\in\calX_{f,i}$, and
\begin{equation}
V_{f,i}(f_i(x_i,\kappa_i(x_i)))-V_{f,i}(x_i)
\leq-\ell_i^0(x_i,\kappa_i(x_i))
\label{eq:terminal_decrease}
\end{equation}
for $x_i\in\calX_{f,i}$. For every $(x_1,\ldots,x_N)$ in the product terminal set and
every edge,
\[
g_{ij}^q(f_i(x_i,\kappa_i(x_i)),f_j(x_j,\kappa_j(x_j)))
\leq-\epsilon_{ij}^q.
\]
\label{ass:terminal}
\end{assumption}

These are local assumptions, not global stabilization claims. In particular,
Assumption~\ref{ass:regularity} specifies the region in which the logarithm and all
constants are valid. Section~\ref{sec:terminal_verification} gives reproducible sampled
evidence for Assumption~\ref{ass:terminal} in the selected terminal set.

\begin{lemma}
\label{lem:cap}
Suppose the committed shifted packets have nonnegative reserves
$r_{ij,k}^q(t)\geq0$ in \eqref{eq:raw_reserve}.
Then the budget selected by \eqref{eq:budget_cap}--\eqref{eq:announced_budget} obeys
\begin{equation}
g_{ij}^q(\bar x_i(k|t),\bar x_j(k|t))
\leq-2L_{ij}^q\rho_j(t)-\epsilon_{ij}^q.
\label{eq:cap_implication}
\end{equation}
Consequently, the committed own trajectory has zero trust distance and satisfies every
tightened edge constraint of $\mathcal P_i(t)$.
\end{lemma}

\begin{proof}
Fix $(i,j,k,q)$. If $L_{ij}^q>0$, then
\begin{align*}
0\leq\rho_j(t)
&\leq\rho_j^{\capb}(t)
\leq\frac{r_{ij,k}^q(t)}{2L_{ij}^q},\\
2L_{ij}^q\rho_j(t)
&\leq-g_{ij}^q(\bar x_i(k|t),\bar x_j(k|t))
-\epsilon_{ij}^q,
\end{align*}
where the first line follows from \eqref{eq:budget_cap}--\eqref{eq:announced_budget}
and the second from \eqref{eq:raw_reserve}. Rearrangement gives
\eqref{eq:cap_implication}. If $L_{ij}^q=0$, the hypothesis directly gives the same
inequality because its budget term vanishes. Taking $X_i=\bar X_i(t)$ therefore
satisfies \eqref{eq:tightened}, while
$d_{H,i}^c(\bar X_i(t),\bar X_i(t))=0\leq\rho_i(t)$ satisfies
\eqref{eq:trust}.
\end{proof}

\begin{proposition}
\label{prop:budget_projection}
For fixed committed packets, let
\begin{equation}
\mathcal B_j(t):=[0,\infty)\cap
\bigcap_{\substack{i:j\in\calN_i,\,1\leq k\leq H\\
q:L_{ij}^q>0}}
\left[0,\frac{r_{ij,k}^q(t)}{2L_{ij}^q}\right].
\label{eq:budget_feasible_interval}
\end{equation}
If all raw reserves are nonnegative, then
$\mathcal B_j(t)=[0,\rho_j^{\capb}(t)]$, and the announced budget is the unique solution of
\begin{equation}
\rho_j(t)=\arg\min_{\rho\in\mathcal B_j(t)}
\frac{1}{2}\bigl(\rho-\rho_j^{\data}(t)\bigr)^2.
\label{eq:budget_projection_program}
\end{equation}
For $\rho_j^{\data}(t)\geq0$, the intervention caused by certification is exactly
\begin{equation}
\rho_j^{\data}(t)-\rho_j(t)
=\bigl[\rho_j^{\data}(t)-\rho_j^{\capb}(t)\bigr]_+.
\label{eq:budget_projection_loss}
\end{equation}
Here $[a-\infty]_+=0$ for an unconstrained budget direction.
\end{proposition}

\begin{proof}
Each two-layer edge inequality with $L_{ij}^q>0$ is equivalent to one interval in
\eqref{eq:budget_feasible_interval}. The intersection of intervals sharing the lower
endpoint zero is
\begin{align*}
\mathcal B_j(t)
&=\left[0,
\min_{i,k,q:L_{ij}^q>0}\frac{r_{ij,k}^q(t)}{2L_{ij}^q}\right]\\
&=[0,\rho_j^{\capb}(t)].
\end{align*}
The objective in \eqref{eq:budget_projection_program} is strictly convex, so its unique
minimizer is the Euclidean projection
$\Pi_{[0,\rho_j^{\capb}(t)]}(\rho_j^{\data}(t))$, which is precisely
\eqref{eq:announced_budget}. Evaluating the projection for a nonnegative request gives
\eqref{eq:budget_projection_loss}. Thus the cap changes the data request only by the
minimum scalar amount required by simultaneous edge feasibility.
\end{proof}

\subsection{Verifiable pairwise certificate}

\begin{theorem}
Suppose \eqref{eq:trust}--\eqref{eq:two_layer} hold for all accepted packets. Then, for
every directed edge, stage, and coupling constraint,
\begin{align}
g_{ij}^q(\widehat x_i(k|t),\widehat x_j(k|t))
&\leq-2L_{ij}^q\rho_j(t)-\epsilon_{ij}^q+L_{ij}^q\mu_j(t)\nonumber\\
&=-\epsilon_{ij}^q-L_{ij}^q\rho_j(t)
-L_{ij}^q\bigl(\rho_j(t)-\mu_j(t)\bigr)\leq0.
\label{eq:accepted_residual}
\end{align}
Under Assumption~\ref{ass:execution}, the actual pairwise constraints at time $t+1$
are hard satisfied. Moreover, each neighbor can independently verify the premise
$d_{H,j}^c(\widehat X_j,\bar X_j)\leq\rho_j$ from the two stored packets.
\label{thm:pairwise}
\end{theorem}

\begin{retained}
\begin{proof}
Fix a directed edge $(i,j)$, stage $k$, and index $q$. By
\eqref{eq:lipschitz_constraint},
\begin{align*}
g_{ij}^q(\widehat x_i,\widehat x_j)
&\leq g_{ij}^q(\widehat x_i,\bar x_j)
+L_{ij}^q d_j^c(\widehat x_j,\bar x_j)\\
&\leq g_{ij}^q(\widehat x_i,\bar x_j)
+L_{ij}^q d_{H,j}^c(\widehat X_j,\bar X_j)\\
&\leq-2L_{ij}^q\rho_j-\epsilon_{ij}^q
+L_{ij}^q\mu_j,
\end{align*}
where the last line uses \eqref{eq:tightened}--\eqref{eq:two_layer} and
\eqref{eq:measured_update}. Retaining the measured value before using
$\mu_j\leq\rho_j$ gives
\[
-2L_{ij}^q\rho_j-\epsilon_{ij}^q+L_{ij}^q\mu_j
=-\epsilon_{ij}^q-L_{ij}^q\rho_j-L_{ij}^q(\rho_j-\mu_j),
\]
which is \eqref{eq:accepted_residual}. At $k=1$,
Assumption~\ref{ass:execution} gives, for both incident agents,
$x_i(t+1)=\widehat x_i(1|t)$ and
$x_j(t+1)=\widehat x_j(1|t)$; hence
\[
g_{ij}^q(x_i(t+1),x_j(t+1))
=g_{ij}^q(\widehat x_i(1|t),\widehat x_j(1|t))\leq0.
\]
Finally, the receiver stores $\bar X_j(t)$ and receives $\widehat X_j(t)$, so every
stage distance and their maximum in \eqref{eq:trajectory_metric} are computable without
access to a private decision variable.
\end{proof}
\end{retained}

\begin{corollary}
\label{cor:network}
Stack all directed coupling functions into
$G(\widehat X(t))=\col\{g_{ij}^q(\widehat x_i(k|t),
\widehat x_j(k|t))\}$. Under Theorem~\ref{thm:pairwise},
\begin{align}
G(\widehat X(t))
&\leq-\col\{\epsilon_{ij}^q+L_{ij}^q\rho_j(t)
+L_{ij}^q(\rho_j(t)-\mu_j(t))\}\nonumber\\
&\leq-\col\{L_{ij}^q\rho_j(t)+\epsilon_{ij}^q\}\leq0
\label{eq:network_certificate}
\end{align}
componentwise. Thus validity of the scalar packet checks composes over the graph without
a centralized trajectory optimization.
\end{corollary}

\begin{proof}
Apply \eqref{eq:accepted_residual} independently to every directed edge, prediction
stage, and constraint index, and stack the resulting scalar inequalities. Componentwise
ordering is preserved by stacking, which gives \eqref{eq:network_certificate}. No
centralized decision variable is introduced because each component uses only its incident
accepted packet pair.
\end{proof}

\begin{corollary}
\label{cor:nonideal}
Let
$\nu_i(t+1)=d_i^c(x_i(t+1),\widehat x_i(1|t))$. If $g_{ij}^q$ is also Lipschitz in
its first argument with constant $L_{ij}^{q,i}$, then
\begin{align}
g_{ij}^q(x_i(t+1),x_j(t+1))
&\leq-\epsilon_{ij}^q-L_{ij}^q\rho_j
-L_{ij}^q(\rho_j-\mu_j)\nonumber\\
&\quad+L_{ij}^{q,i}\nu_i+L_{ij}^q\nu_j.
\label{eq:nonideal_residual}
\end{align}
The a posteriori execution margin carried by the accepted packet is
\begin{equation}
m_{ij}^q(t):=\epsilon_{ij}^q+L_{ij}^q\rho_j(t)
+L_{ij}^q\bigl(\rho_j(t)-\mu_j(t)\bigr).
\label{eq:execution_margin}
\end{equation}
Hence the true-state constraint is preserved whenever
\begin{equation}
L_{ij}^{q,i}\nu_i(t+1)+L_{ij}^q\nu_j(t+1)\leq m_{ij}^q(t).
\label{eq:execution_margin_test}
\end{equation}
Consequently, if known bounds $\nu_i\leq\bar\nu_i$ are included in the tightening as
$L_{ij}^{q,i}\bar\nu_i+L_{ij}^q\bar\nu_j$ and the resulting hard problem is feasible,
then the true state satisfies
$g_{ij}^q(x_i(t+1),x_j(t+1))\leq0$.
\end{corollary}

\begin{proof}
Set $\widehat x_i^+=\widehat x_i(1|t)$ and
$\widehat x_j^+=\widehat x_j(1|t)$. Two applications of the Lipschitz inequality give
\begin{align*}
g_{ij}^q(x_i(t+1),x_j(t+1))
&\leq g_{ij}^q(\widehat x_i^+,\widehat x_j^+)
+L_{ij}^{q,i}d_i^c(x_i(t+1),\widehat x_i^+)\\
&\quad+L_{ij}^q d_j^c(x_j(t+1),\widehat x_j^+)\\
&\leq-\epsilon_{ij}^q-L_{ij}^q\rho_j
-L_{ij}^q(\rho_j-\mu_j)
+L_{ij}^{q,i}\nu_i+L_{ij}^q\nu_j,
\end{align*}
which proves \eqref{eq:nonideal_residual}. If
$L_{ij}^{q,i}\nu_i+L_{ij}^q\nu_j\leq m_{ij}^q(t)$, substitution of
\eqref{eq:execution_margin} proves the hard inequality directly. If only prior bounds
are available and $\nu_\ell\leq\bar\nu_\ell$, replace \eqref{eq:two_layer} by
\[
\eta_{ij}^{q,\mathrm{rob}}
=2L_{ij}^q\rho_j+\epsilon_{ij}^q
+L_{ij}^{q,i}\bar\nu_i+L_{ij}^q\bar\nu_j.
\]
Repeating the preceding nominal calculation gives
$g_{ij}^q(\widehat x_i^+,\widehat x_j^+)
\leq-\epsilon_{ij}^q-L_{ij}^q\rho_j-L_{ij}^q(\rho_j-\mu_j)
-L_{ij}^{q,i}\bar\nu_i-L_{ij}^q\bar\nu_j$.
Substitution in the first inequality cancels both error terms and yields
$g_{ij}^q(x_i(t+1),x_j(t+1))\leq
-\epsilon_{ij}^q-L_{ij}^q\rho_j-L_{ij}^q(\rho_j-\mu_j)\leq0$.
Bounded packet age may be converted into $\bar\nu_i$ through a model-dependent state
increment bound; unbounded delay and asynchronous execution are outside this result.
\end{proof}

\subsection{Recursive feasibility by budget projection}

\begin{theorem}
\label{thm:recursive}
Let Assumptions~\ref{ass:regularity}--\ref{ass:terminal} hold. Suppose the time-$0$
shifted candidate satisfies the hard constraints with the offsets
$\epsilon_{ij}^q$. If budgets are selected by \eqref{eq:budget_cap}--
\eqref{eq:announced_budget}, then the hard local problems are recursively feasible. The
shifted candidate is an admissible fallback at every sampling instant, and the actual
nominal pairwise constraints remain satisfied.
\end{theorem}

\begin{retained}
\begin{proof}
The time-$0$ candidate is feasible by hypothesis. For the induction step, assume that
the time-$t$ hard problems are feasible and let
$\widehat U_i(t)=\{\widehat u_i(0|t),\ldots,\widehat u_i(H-1|t)\}$ denote the input
sequence associated with the accepted state trajectory $\widehat X_i(t)$. Define
the time-$(t+1)$ candidate by deleting the applied input and appending the terminal
feedback:
\begin{align}
\widetilde U_i(t+1)&=\{\widehat u_i(1|t),\ldots,
\widehat u_i(H-1|t),\kappa_i(\widehat x_i(H|t))\},
\label{eq:shifted_input}\\
\widetilde x_i(k|t+1)&=\widehat x_i(k+1|t),\quad 0\leq k<H,\nonumber\\
\widetilde x_i(H|t+1)&=f_i(\widehat x_i(H|t),
\kappa_i(\widehat x_i(H|t))).
\label{eq:shifted_state}
\end{align}
By construction, $\widetilde X_i(t+1)=\bar X_i(t+1)$. Assumption~\ref{ass:execution}
gives $\widetilde x_i(0|t+1)=\widehat x_i(1|t)=x_i(t+1)$, so the initial condition of
the next local problem is satisfied. For $0\leq k\leq H-2$,
\begin{align*}
\widetilde x_i(k+1|t+1)
&=\widehat x_i(k+2|t)\\
&=f_i(\widehat x_i(k+1|t),\widehat u_i(k+1|t))\\
&=f_i(\widetilde x_i(k|t+1),\widetilde u_i(k|t+1)),
\end{align*}
and the same equality at $k=H-1$ follows from the last line of
\eqref{eq:shifted_state}. Thus the dynamics hold at every stage. For
$0\leq k<H$, $\widetilde x_i(k|t+1)=\widehat x_i(k+1|t)$, and for
$0\leq k<H-1$, $\widetilde u_i(k|t+1)=\widehat u_i(k+1|t)$; their hard state and
input constraints are inherited from the accepted solution. Assumption~\ref{ass:terminal}
gives
\[
\kappa_i(\widehat x_i(H|t))\in\calU_i,\qquad
\widetilde x_i(H|t+1)\in\calX_{f,i}\subseteq\calX_i\cap\calD_i,
\]
so the appended input, state, and terminal constraint are also admissible.

It remains to verify the coupled constraints. For $1\leq k<H$,
Theorem~\ref{thm:pairwise} at old stage $k+1$ gives
\begin{equation}
g_{ij}^q(\widetilde x_i(k|t+1),\widetilde x_j(k|t+1))
\leq-L_{ij}^q\rho_j(t)-\epsilon_{ij}^q
\leq-\epsilon_{ij}^q.
\label{eq:inherited_edge}
\end{equation}
For $k=H$, the terminal product condition in Assumption~\ref{ass:terminal} gives the
same final inequality. Hence, for every constrained $(i,j,k,q)$,
\[
r_{ij,k}^q(t+1)
=-g_{ij}^q(\bar x_i(k|t+1),\bar x_j(k|t+1))
-\epsilon_{ij}^q\geq0.
\]
Consequently every reserve
$r_{ij,k}^q(t+1)$ in \eqref{eq:raw_reserve} is nonnegative. Applying
Lemma~\ref{lem:cap} with the newly announced budget gives
\begin{equation}
g_{ij}^q(\bar x_i(k|t+1),\bar x_j(k|t+1))
\leq-2L_{ij}^q\rho_j(t+1)-\epsilon_{ij}^q.
\label{eq:next_tightened_shift}
\end{equation}
If agent $i$ chooses its committed trajectory as the decision trajectory, then
\eqref{eq:next_tightened_shift} is exactly the tightened edge constraint and
$d_{H,i}^c(\bar X_i(t+1),\bar X_i(t+1))=0$ satisfies the trust region. Thus a hard-feasible
point satisfying the dynamics, local constraints, terminal constraint, trust region,
and every coupling constraint is known before the time-$(t+1)$ numerical optimization.

The argument is based on the accepted trajectory, not on an unchecked optimizer output.
If the optimizer fails any hard numerical test, the committed shifted candidate is
accepted and its first input is applied, so the same construction remains valid.
Thus feasibility at $t$ implies feasibility at $t+1$; induction from the base case proves
recursive feasibility for all $t\geq0$. Finally, Assumption~\ref{ass:execution} and
Theorem~\ref{thm:pairwise} at $k=1$ give
$g_{ij}^q(x_i(t+1),x_j(t+1))\leq0$ for every actual nominal edge.
\end{proof}
\end{retained}

\begin{corollary}
\label{cor:packet_invariance}
Let $\mathfrak C$ contain all pairs $(\bar X,\rho)$ for which every $\bar X_i$ satisfies
the shifted local and terminal conditions and
\begin{equation}
g_{ij}^q(\bar x_i(k),\bar x_j(k))
+2L_{ij}^q\rho_j+\epsilon_{ij}^q\leq0
\label{eq:certified_prediction_set}
\end{equation}
for every directed edge, stage, and constraint. Under the conditions of
Theorem~\ref{thm:recursive}, the accepted-trajectory, terminal-extension, and budget-cap
update maps $\mathfrak C$ into itself. Thus $\mathfrak C$ is positively invariant for
the augmented closed-loop state consisting of the plant state and committed prediction
packets.
\end{corollary}

\begin{proof}
Take $(\bar X(t),\rho(t))\in\mathfrak C$. Equation
\eqref{eq:certified_prediction_set}, zero trust distance, and the local and terminal
conditions make $\bar X_i(t)$ feasible for every local problem. After the hard
acceptance check, Theorem~\ref{thm:pairwise} gives
\[
g_{ij}^q(\widehat x_i(k|t),\widehat x_j(k|t))
\leq-L_{ij}^q\rho_j(t)-\epsilon_{ij}^q
\leq-\epsilon_{ij}^q.
\]
Shifting the accepted trajectories preserves this inequality at inherited stages, while
Assumption~\ref{ass:terminal} supplies it at the appended stage. Therefore
$r_{ij,k}^q(t+1)\geq0$ for every index. By
\eqref{eq:budget_cap}--\eqref{eq:announced_budget},
$2L_{ij}^q\rho_j(t+1)\leq r_{ij,k}^q(t+1)$, and hence
\[
g_{ij}^q(\bar x_i(k|t+1),\bar x_j(k|t+1))
+2L_{ij}^q\rho_j(t+1)+\epsilon_{ij}^q\leq0.
\]
The local and terminal conditions are preserved by the same shift construction, and
$\bar x_i(0|t+1)=x_i(t+1)$ by nominal execution. Thus
$(\bar X(t+1),\rho(t+1))\in\mathfrak C$.
\end{proof}

The shifted-feasibility margin used below is
\begin{equation}
\mathcal R(t)=\min_{i,j,k,q}
\{-g_{ij}^q(\bar x_i(k|t),\bar x_j(k|t))
-2L_{ij}^q\rho_j(t)-\epsilon_{ij}^q\}.
\label{eq:shift_feasibility_margin}
\end{equation}
The shifted candidate passes exactly when $\mathcal R(t)\geq0$, together with the local
state, input, trust, and terminal checks. A fixed budget that is not projected through
\eqref{eq:budget_cap} can have $\mathcal R(t)<0$ even if a new optimizer happens to be
found.

\subsection{Cost regularity and practical decrease}

\begin{lemma}
\label{lem:cost_lipschitz}
Let $j=i$ represent an own-shift regularizer if one is used. Suppose each packet-dependent stage term
$\ell_{ij,k}(x_i,u_i,x_j)$, $j\in\calN_i^+$, is continuously differentiable on the compact set
$\calD\times\calU_i\times\calD$. Then, for fixed own state and input trajectories,
\begin{equation}
|J_i(U_i;x_i,Y_{\calN_i^+})-J_i(U_i;x_i,Z_{\calN_i^+})|
\leq\sum_{j\in\calN_i^+}L_{ij}^{J}d_{H,j}(Y_j,Z_j),
\label{eq:cost_lipschitz}
\end{equation}
where one valid explicit choice is
\begin{equation}
L_{ij}^{J}=\sum_{k=0}^{H-1}
\sup_{\calD\times\calU_i\times\calD}
\|D_{x_j}\ell_{ij,k}\|_{(d_j)^*}.
\end{equation}
For the quadratic relative terms in \eqref{eq:stage_cost}, let
$\bar e_{ij}^{p}=\sup_{\calD}\|e_{ij}^{p}\|$,
$\bar e_{ij}^{R}=\sup_{\calD}\|e_{ij}^{R}\|$, and
\begin{equation}
c_{\log}(\bar\theta)=
\sup_{\|\phi\|\leq\bar\theta<\pi}\|J_r^{-1}(\phi)\|.
\label{eq:log_jacobian_bound}
\end{equation}
Here $J_r^{-1}$ is the inverse right Jacobian of the $\SO$ logarithm.
Direct differentiation gives the computable stage bounds
\begin{align}
M_{ij,k}^{p}&\leq2\|Q_{ij}^{p}\|\bar e_{ij}^{p},\nonumber\\
M_{ij,k}^{R}&\leq
2\|Q_{ij}^{R}\|\bar e_{ij}^{R}c_{\log}(\bar\theta),
\qquad
M_{ij,k}\leq M_{ij,k}^{p}+M_{ij,k}^{R}.
\label{eq:quadratic_cost_constants}
\end{align}
Neighbor-dependent terminal terms, if present, add their corresponding supremum.
\end{lemma}

\begin{proof}
Fix $i,j,k$ and put
$M_{ij,k}=\sup\|D_{x_j}\ell_{ij,k}\|_{(d_j)^*}$. For any $\delta>0$, choose
an absolutely continuous local curve $\gamma:[0,1]\to\calD$ from $Z_j(k)$ to
$Y_j(k)$ whose $d_j$-length is at most
$d_j(Y_j(k),Z_j(k))+\delta$. The fundamental theorem of calculus and dual-norm
inequality give
\begin{align*}
&|\ell_{ij,k}(x_i,u_i,Y_j(k))
-\ell_{ij,k}(x_i,u_i,Z_j(k))|\\
&\quad\leq\int_0^1
\|D_{x_j}\ell_{ij,k}(x_i,u_i,\gamma(s))\|_{(d_j)^*}
\|\dot\gamma(s)\|_{d_j}\,ds\\
&\quad\leq M_{ij,k}
\bigl(d_j(Y_j(k),Z_j(k))+\delta\bigr).
\end{align*}
Letting $\delta\downarrow0$, summing over $k$ and $j$, and using
$d_j(Y_j(k),Z_j(k))\leq d_{H,j}(Y_j,Z_j)$ proves
\eqref{eq:cost_lipschitz}. Compactness and continuous differentiability make every
$M_{ij,k}$ finite. A packet-dependent terminal term is handled by the identical
calculation at $k=H$.
\end{proof}

\begin{lemma}
\label{lem:shift_regular}
Let $S_i$ denote the shift-and-terminal-extension map defined by
\eqref{eq:shifted_input}--\eqref{eq:shifted_state}. On the compact local operating set,
there is a finite $c_{S,i}>0$ such that
\begin{equation}
d_{H,i}(S_iY_i,S_iZ_i)\leq c_{S,i}d_{H,i}(Y_i,Z_i).
\label{eq:shift_lipschitz}
\end{equation}
\end{lemma}

\begin{proof}
Define $F_i(x)=f_i(x,\kappa_i(x))$. In local coordinates, local Lipschitz
continuity on the compact chart gives finite constants
$L_{f,x}^i$, $L_{f,u}^i$, and $L_\kappa^i$ such that
\begin{align*}
d_i(f_i(x,u),f_i(y,v))
&\leq L_{f,x}^i d_i(x,y)+L_{f,u}^i\|u-v\|,\\
\|\kappa_i(x)-\kappa_i(y)\|&\leq L_\kappa^i d_i(x,y).
\end{align*}
For the projected LQR law, nonexpansiveness of the box projection and local metric
equivalence give a finite choice proportional to
$\|K_i\|\operatorname{Lip}(e_i)$. Hence
\begin{equation}
d_i(F_i(y),F_i(z))
\leq\bigl(L_{f,x}^i+L_{f,u}^iL_\kappa^i\bigr)d_i(y,z)
=:L_{F,i}d_i(y,z).
\label{eq:terminal_map_derivative}
\end{equation}
This argument also covers projection breakpoints and does not require differentiability
of $\kappa_i$ there.
For $0\leq k<H$, $(S_iY_i)(k)=Y_i(k+1)$, while
$(S_iY_i)(H)=F_i(Y_i(H))$. Consequently,
\begin{align*}
d_{H,i}(S_iY_i,S_iZ_i)
&=\max\!\left\{
\max_{0\leq k<H}d_i(Y_i(k+1),Z_i(k+1)),\right.\\
&\hspace{7.2em}\left.
d_i(F_i(Y_i(H)),F_i(Z_i(H)))\right\}\\
&\leq\max\{1,L_{F,i}\}\,d_{H,i}(Y_i,Z_i).
\end{align*}
Thus \eqref{eq:shift_lipschitz} holds with
$c_{S,i}=\max\{1,L_{F,i}\}$.
\end{proof}

\begin{assumption}
\label{ass:value_terminal}
The same stage-cost functions and weights are used at every time and horizon index.
Let $z=(z_1,\ldots,z_N)$ collect accepted terminal states and let
$y=(y_1,\ldots,y_N)$ collect the committed terminal packet states. For every terminal
pair $(z,y)$ encountered in the certified local set, with
$z_i,y_i\in\calX_{f,i}$, the terminal feedback and value satisfy
\begin{equation}
\sum_{i=1}^{N}\!\left[V_{f,i}(F_i(z_i))-V_{f,i}(z_i)
+\ell_i^{T}(z_i,\kappa_i(z_i),y_{\calN_i^+})\right]\leq0,
\label{eq:network_terminal_decrease}
\end{equation}
where $F_i(z_i)=f_i(z_i,\kappa_i(z_i))$ and $\ell_i^{T}$ is the complete appended
stage generated by the shifted packet, including own-shift and coordination terms.
The two terminal vectors are kept distinct because the shifted candidate propagates
$\widehat x_i(H|t)$ while its packet-dependent terms use $\bar x_j(H|t)$. Slew
variables may be included in the local Euclidean state. This condition is local and is
required only for the value-decrease result, not for
Theorems~\ref{thm:pairwise}--\ref{thm:recursive}.
\end{assumption}

Combining Lemma~\ref{lem:cost_lipschitz}, the shifted-sequence argument, and
Assumption~\ref{ass:value_terminal} gives the following local statement.

\begin{theorem}
\label{thm:decrease}
Under the assumptions of Theorem~\ref{thm:recursive} and
Assumption~\ref{ass:value_terminal}, suppose the nominal network
stage cost is lower bounded by a positive-definite function
$\alpha(\|e(t)\|)$. Then finite constants $C_j\geq0$ exist such that
\begin{equation}
V(t+1)-V(t)
\leq-\alpha(\|e(t)\|)+\sum_{j=1}^{N}C_j\chi_j(t),
\label{eq:practical_decrease}
\end{equation}
where $V(t)=\sum_i\widehat J_i(t)$ is the sum of the actually accepted hard-feasible
costs in \eqref{eq:accepted_value}. Here $e(t)$ denotes the closed-loop error used in
the comparison bounds; it includes stored packet or slew-state deviations when those
variables are retained in the value function, and the stated stage-cost lower bound is
understood with respect to this same error coordinate. Suppose, on a forward-invariant local
sublevel set,
\[
\underline\alpha_V(\|e\|)\leq V\leq\overline\alpha_V(\|e\|),
\]
where $\underline\alpha_V$, $\overline\alpha_V$, and $\alpha$ are
class-$\mathcal K_\infty$ functions on the relevant range. Then, for every integer
$T\geq1$,
\begin{align}
\sum_{t=0}^{T-1}\alpha(\|e(t)\|)
&\leq V(0)-V(T)
+\sum_{j=1}^{N}C_j\sum_{t=0}^{T-1}\chi_j(t)\nonumber\\
&\leq V(0)+\sum_{j=1}^{N}C_j\sum_{t=0}^{T-1}\chi_j(t).
\label{eq:finite_horizon_dissipation}
\end{align}
Consequently,
\begin{equation}
\frac{1}{T}\sum_{t=0}^{T-1}\alpha(\|e(t)\|)
\leq\frac{V(0)}{T}
+\sum_{j=1}^{N}C_j\frac{1}{T}\sum_{t=0}^{T-1}\chi_j(t).
\label{eq:average_tracking_bound}
\end{equation}
In particular, $\sum_{t=0}^{\infty}\chi_j(t)<\infty$ for every $j$ implies
$e(t)\to0$.

The additive term in this result uses the full-state displacement $\chi_j$, not the
coupling-output displacement $\mu_j$ used for constraint tightening; no bound
$\chi_j\leq\rho_j$ is assumed. If a design also imposes a full-state trust radius
$d_{H,j}(\widehat X_j,\bar X_j)\leq\rho_j^J$, then one may take
$\bar\chi_j=\sup_t\rho_j^J(t)$. Otherwise $\bar\chi_j$ is obtained from the declared
compact set or from the recorded full trajectories.

If $\chi_j(t)\leq\bar\chi_j$, define $\bar d=\sum_jC_j\bar\chi_j$,
$\psi=\alpha\circ\overline\alpha_V^{-1}$, and
$v_d=\psi^{-1}(\bar d)$. Then
\begin{equation}
\limsup_{t\to\infty}\|e(t)\|
\leq\underline\alpha_V^{-1}(v_d+\bar d).
\label{eq:explicit_ultimate_bound}
\end{equation}
If $\chi_j(t)\to0$, then $e(t)\to0$.
\end{theorem}

\begin{revision}
\begin{proof}
Let $Y_j(t)=\bar X_j(t)$. Theorem~\ref{thm:recursive} supplies the hard-feasible shifted
input $\widetilde U_i(t+1)$ in \eqref{eq:shifted_input}. Since this input belongs to
$\mathcal A_i(t+1)$, \eqref{eq:accepted_value} and addition/subtraction of the cost
parameterized by $S Y(t)$ give
\begin{align}
\widehat J_i(t+1)-\widehat J_i(t)
\leq{}&J_i(\widetilde U_i(t+1);x_i(t+1),
S Y_{\calN_i^+}(t))-\widehat J_i(t)\nonumber\\
&+\Delta_{J,i}(t),
\label{eq:value_split}
\end{align}
where the packet-parameter perturbation is
\begin{align}
\Delta_{J,i}(t)={}&J_i(\widetilde U_i(t+1);x_i(t+1),
\bar X_{\calN_i^+}(t+1))\nonumber\\
&-J_i(\widetilde U_i(t+1);x_i(t+1),
S Y_{\calN_i^+}(t)).
\label{eq:value_packet_perturbation}
\end{align}
The input and own-state trajectories are identical in the two terms defining
$\Delta_{J,i}$; only the packet parameters change. Since
$\bar X_j(t+1)=S_j\widehat X_j(t)$ and $Y_j(t)=\bar X_j(t)$,
\begin{equation}
d_{H,j}(\bar X_j(t+1),S_jY_j(t))
\leq c_{S,j}d_{H,j}(\widehat X_j(t),\bar X_j(t))
=c_{S,j}\chi_j(t).
\label{eq:shift_packet_bound}
\end{equation}
Lemma~\ref{lem:cost_lipschitz} applied to
\eqref{eq:value_packet_perturbation} therefore yields
\begin{equation}
|\Delta_{J,i}(t)|\leq
\sum_{j\in\calN_i^+}L_{ij}^{J}c_{S,j}\chi_j(t).
\label{eq:value_perturbation}
\end{equation}

For an explicit cancellation, define
\begin{equation*}
\widehat\ell_i^t(k):=\ell_i(\widehat x_i(k|t),\widehat u_i(k|t),
Y_{\calN_i^+}(k|t)),\qquad 0\leq k<H,
\end{equation*}
and let $\ell_i^T(t)$ denote the complete appended stage generated by
$(\widehat x_i(H|t),\kappa_i(\widehat x_i(H|t)),Y_{\calN_i^+}(H|t))$. With
$F_i(x)=f_i(x,\kappa_i(x))$, the two finite sums are
\begin{align}
\widehat J_i(t)
&=\sum_{k=0}^{H-1}\widehat\ell_i^t(k)
+V_{f,i}(\widehat x_i(H|t)),\nonumber\\
J_i(\widetilde U_i(t+1);x_i(t+1),S Y_{\calN_i^+}(t))
&=\sum_{k=1}^{H-1}\widehat\ell_i^t(k)+\ell_i^T(t)
+V_{f,i}(F_i(\widehat x_i(H|t))).
\label{eq:shifted_cost_expansion}
\end{align}
Hence all inherited stages cancel exactly and
\begin{align}
&J_i(\widetilde U_i;x_i(t+1),S Y_{\calN_i^+})-\widehat J_i(t)\nonumber\\
&=-\widehat\ell_i^t(0)+\ell_i^T(t)
+V_{f,i}(F_i(\widehat x_i(H|t)))-V_{f,i}(\widehat x_i(H|t)).
\label{eq:exact_cost_cancellation}
\end{align}
Summing \eqref{eq:exact_cost_cancellation}, applying
\eqref{eq:network_terminal_decrease} with
$z_i=\widehat x_i(H|t)$ and $y_j=Y_j(H|t)$, and using
$\widehat\ell_i^t(0)\geq\ell_i^0(x_i(t),u_i(t))$ gives
\begin{align}
&\sum_i\!\left[J_i(\widetilde U_i;x_i(t+1),S Y_{\calN_i^+})
-\widehat J_i(t)\right]\nonumber\\
&\hspace{3em}\leq-\sum_i\ell_i^0(x_i(t),u_i(t))
\leq-\alpha(\|e(t)\|).
\label{eq:nominal_shift_decrease}
\end{align}
Combining \eqref{eq:value_split}, \eqref{eq:value_perturbation}, and
\eqref{eq:nominal_shift_decrease}, then exchanging the two finite sums, proves
\eqref{eq:practical_decrease} with
\begin{equation}
C_j=c_{S,j}\sum_{i:j\in\calN_i^+}L_{ij}^{J}.
\label{eq:explicit_Cj}
\end{equation}

Summing \eqref{eq:practical_decrease} from $t=0$ to $T-1$ telescopes the value terms:
\begin{align*}
\sum_{t=0}^{T-1}\alpha(\|e(t)\|)
&\leq V(0)-V(T)+\sum_{t=0}^{T-1}\sum_{j=1}^{N}C_j\chi_j(t)\\
&\leq V(0)+\sum_{j=1}^{N}C_j\sum_{t=0}^{T-1}\chi_j(t),
\end{align*}
because $V(T)\geq\underline\alpha_V(\|e(T)\|)\geq0$. This proves
\eqref{eq:finite_horizon_dissipation}; division by $T$ gives
\eqref{eq:average_tracking_bound}. If every full-state packet-displacement sequence is summable, the
right-hand side of \eqref{eq:finite_horizon_dissipation} remains finite as
$T\to\infty$. Hence $\sum_t\alpha(\|e(t)\|)<\infty$, which implies
$\alpha(\|e(t)\|)\to0$ and therefore $e(t)\to0$.

For the stability conclusion, put $d_t=\sum_jC_j\chi_j(t)$. Since
$V(t)\leq\overline\alpha_V(\|e(t)\|)$,
\[
\alpha(\|e(t)\|)
\geq\alpha(\overline\alpha_V^{-1}(V(t)))=\psi(V(t)),
\]
and \eqref{eq:practical_decrease} implies
\begin{equation}
V(t+1)\leq V(t)-\psi(V(t))+d_t.
\label{eq:scalar_comparison}
\end{equation}
If $d_t\leq\bar d$, fix $\delta>0$ and define
\begin{equation*}
\varepsilon_\delta:=\psi(v_d+\delta)-\bar d>0,
\qquad T_\delta:=\inf\{t\geq0:V(t)\leq v_d+\delta\}.
\end{equation*}
For $0\leq t<T_\delta$, \eqref{eq:scalar_comparison} gives
\begin{equation}
V(t+1)\leq V(t)-\varepsilon_\delta,\qquad
T_\delta\leq
\left\lceil\frac{[V(0)-v_d-\delta]_+}{\varepsilon_\delta}\right\rceil.
\label{eq:hitting_time}
\end{equation}
Put $a_\delta=v_d+\delta$, $b_\delta=a_\delta+\bar d$. The interval
$\mathcal I_\delta=[0,b_\delta]$ is positively invariant because
\begin{equation}
V(t+1)\leq
\begin{cases}
V(t)+\bar d\leq b_\delta,&V(t)\leq a_\delta,\\
V(t)-\varepsilon_\delta\leq b_\delta,&a_\delta<V(t)\leq b_\delta.
\end{cases}
\label{eq:comparison_invariance}
\end{equation}
Thus $\limsup_tV(t)\leq v_d+\bar d$ after letting $\delta\downarrow0$, and the lower
comparison bound proves \eqref{eq:explicit_ultimate_bound}. If $d_t\to0$, set
$\bar d_T=\sup_{t\geq T}d_t$. Applying the same calculation on every tail gives
\begin{equation*}
\limsup_{t\to\infty}V(t)
\leq\psi^{-1}(\bar d_T)+\bar d_T\quad\forall T,
\qquad \bar d_T\downarrow0,
\end{equation*}
so $V(t)\to0$ and $e(t)\to0$.
No slack term appears because \eqref{eq:hard_terminal} and all certified pairwise
constraints are hard.
\end{proof}
\end{revision}

\begin{corollary}
\label{cor:terminal_residual}
If the complete network terminal inequality in
Assumption~\ref{ass:value_terminal} is not imposed pointwise, define the directly
evaluable terminal-shift residual
\begin{align}
r_T(t):=\sum_{i=1}^{N}\bigl[&\ell_i^T(t)
+V_{f,i}(F_i(\widehat x_i(H|t)))\nonumber\\
&-V_{f,i}(\widehat x_i(H|t))\bigr],
\qquad \zeta(t):=[r_T(t)]_+.
\label{eq:terminal_shift_residual}
\end{align}
Let $d_J(t):=\sum_i|\Delta_{J,i}(t)|$. The candidate decomposition gives
\begin{align}
V(t+1)-V(t)
&\leq-\alpha(\|e(t)\|)+d_J(t)+\zeta(t)\nonumber\\
&\leq-\alpha(\|e(t)\|)+\sum_{j=1}^{N}C_j\chi_j(t)+\zeta(t).
\label{eq:practical_decrease_residual}
\end{align}
All accumulated and ultimate-bound conclusions of
Theorem~\ref{thm:decrease} hold after replacing
$d_t=\sum_jC_j\chi_j(t)$ by either $d_J(t)+\zeta(t)$ or
$\sum_jC_j\chi_j(t)+\zeta(t)$. Under
Assumption~\ref{ass:value_terminal}, $\zeta(t)=0$ and
\eqref{eq:practical_decrease_residual} reduces to
\eqref{eq:practical_decrease}.
\end{corollary}

\begin{proof}
Before invoking \eqref{eq:network_terminal_decrease}, summing
\eqref{eq:exact_cost_cancellation} gives the dropped first stage plus $r_T(t)$.
Using $r_T(t)\leq\zeta(t)$ and
$\sum_i\Delta_{J,i}(t)\leq d_J(t)$ proves the first inequality in
\eqref{eq:practical_decrease_residual}; the packet bound
\eqref{eq:value_perturbation} proves the second. The remaining arguments depend only on the
nonnegative additive sequence $d_t$, so the same telescoping and comparison proofs
apply.
\end{proof}

\begin{corollary}
\label{cor:ultimate_bound}
Under Theorem~\ref{thm:decrease}, suppose $\chi_j(t)\leq\bar\chi_j$. If a desired
radius $r_b$ in the local comparison region satisfies
\begin{equation}
\psi^{-1}\!\left(\sum_{j=1}^{N}C_j\bar\chi_j\right)
+\sum_{j=1}^{N}C_j\bar\chi_j
\leq\underline\alpha_V(r_b),
\label{eq:budget_to_bound}
\end{equation}
then $\limsup_{t\to\infty}\|e(t)\|\leq r_b$ for trajectories that remain in that
local sublevel set.
\end{corollary}

\begin{proof}
Let $\bar d=\sum_jC_j\bar\chi_j$. Theorem~\ref{thm:decrease} gives
$\limsup_tV(t)\leq\psi^{-1}(\bar d)+\bar d$. Since
$\underline\alpha_V(\|e(t)\|)\leq V(t)$ and
$\underline\alpha_V^{-1}$ is increasing,
\[
\limsup_{t\to\infty}\|e(t)\|
\leq\underline\alpha_V^{-1}
\!\left(\psi^{-1}(\bar d)+\bar d\right)
\leq r_b,
\]
where the last inequality is exactly \eqref{eq:budget_to_bound}.
\end{proof}

\begin{corollary}
\label{cor:quadratic_budget_bound}
Suppose the local comparison bounds in Theorem~\ref{thm:decrease} can be chosen as
\begin{equation}
\underline c\|e\|^2\leq V\leq\overline c\|e\|^2,
\qquad
\sum_i\ell_i^0(x_i,u_i)\geq a\|e\|^2
\label{eq:quadratic_comparison_bounds}
\end{equation}
with $\underline c,\overline c,a>0$. For
$\bar d=\sum_jC_j\bar\chi_j$, one has
\begin{align}
\Delta V(t)&\leq-a\|e(t)\|^2+\bar d,
\label{eq:quadratic_practical_decrease}\\
\limsup_{t\to\infty}\|e(t)\|
&\leq
\sqrt{\frac{1+\overline c/a}{\underline c}\,\bar d}.
\label{eq:quadratic_ultimate_bound}
\end{align}
Hence $V$ decreases strictly outside $\|e\|\leq\sqrt{\bar d/a}$, and the sufficient
full-packet displacement condition
\begin{equation}
\sum_{j=1}^{N}C_j\bar\chi_j
\leq\frac{\underline c\,a}{a+\overline c}\,r_b^2
\label{eq:quadratic_budget_design}
\end{equation}
is sufficient for $\limsup_t\|e(t)\|\leq r_b$.
\end{corollary}

\begin{proof}
Choose $\alpha(r)=ar^2$, $\underline\alpha_V(r)=\underline c r^2$, and
$\overline\alpha_V(r)=\overline c r^2$. Then
\begin{equation*}
\psi(v)=\alpha\!\left(\overline\alpha_V^{-1}(v)\right)
=\frac{a}{\overline c}v,
\qquad
v_d=\psi^{-1}(\bar d)=\frac{\overline c}{a}\bar d.
\end{equation*}
Equation~\eqref{eq:quadratic_practical_decrease} follows directly from
\eqref{eq:practical_decrease}. Substitution of $v_d$ into
\eqref{eq:explicit_ultimate_bound} gives
\[
\underline\alpha_V^{-1}(v_d+\bar d)
=\sqrt{\frac{(\overline c/a+1)\bar d}{\underline c}},
\]
which proves \eqref{eq:quadratic_ultimate_bound}. Requiring its right-hand side not to
exceed $r_b$ and rearranging yields \eqref{eq:quadratic_budget_design}.
\end{proof}
\end{revision}
\endgroup

\section{Spacecraft Formation Study}

\subsection{Model, geometry, and reproducible constants}

Each spacecraft state is $(p_i,v_i,R_i,\Omega_i)\in\R^3\times\R^3\times\SO\times
\R^3$. With sample time $h$, Hill--Clohessy--Wiltshire acceleration
$a_{\rm HCW}$ \cite{ClohessyWiltshire1960}, diagonal inertia $J$, commanded
acceleration $a_i$, and body torque $\tau_i$, the implemented update is
\begin{align}
v_i^+&=v_i+h(a_{\rm HCW}(p_i,v_i)+a_i),\nonumber\\
p_i^+&=p_i+hv_i^+,\label{eq:hcw}\\
\Omega_i^+&=\Omega_i+hJ^{-1}(\tau_i-\Omega_i\times J\Omega_i),\nonumber\\
R_i^+&=R_i\Exp(h\Omega_i^+).
\label{eq:attitude_update}
\end{align}
Writing $p_i=(p_{x,i},p_{y,i},p_{z,i})$ and
$v_i=(v_{x,i},v_{y,i},v_{z,i})$, the acceleration in \eqref{eq:hcw} is
\begin{equation}
a_{\rm HCW}(p_i,v_i)=
\begin{bmatrix}
3n^2p_{x,i}+2nv_{y,i} & -2nv_{x,i} & -n^2p_{z,i}
\end{bmatrix}^{\!\top}.
\label{eq:hcw_acceleration}
\end{equation}
Angular velocity is expressed in the body frame and multiplies $R_i$ on the right.
The hat map is included in the definition of $\Exp$. Using $\Omega_i^+$ in the group
step follows the same velocity-first semi-implicit Euler ordering as the translation
update in \eqref{eq:hcw}; this ordering is a discretization choice, not an assumption of
the certificate argument. Right multiplication through $\Exp$ preserves
$R_i^\top R_i=I$ exactly in exact arithmetic. The method is first order, not an LGVI;
a higher-order or variational integrator can replace it without changing the metric
certificate argument \cite{NordkvistSanyal2010}.

For $N$ agents, the desired positions are
\begin{equation}
p_{i,d}=r_f\begin{bmatrix}
\cos\phi_i & \sin\phi_i & 0.15\sin(2\phi_i)
\end{bmatrix}^{\!\top},\qquad
\phi_i=\frac{2\pi(i-1)}{N}.
\label{eq:desired_formation}
\end{equation}
Each desired attitude points the body first axis toward the target. More precisely, with
$b_{1,i}=(p_T-p_{i,d})/\|p_T-p_{i,d}\|$, a fixed inertial reference not parallel to
$b_{1,i}$ is projected to construct $b_{3,i}$, then
$b_{2,i}=b_{3,i}\times b_{1,i}$ and
$R_{i,d}=[b_{1,i}\ b_{2,i}\ b_{3,i}]$. The desired relative attitude in
\eqref{eq:gatt} is $R_{ij}^{d}=R_{i,d}^\top R_{j,d}$. This makes the rotational state
enter not only an individual pointing cost but also a hard neighbor-dependent manifold
constraint.

\begin{revision}
Besides the pairwise constraints \eqref{eq:gcol}--\eqref{eq:gatt}, each prediction obeys
\begin{align}
\|\Omega_i\|&\leq1.3,\qquad |p_i|\leq(5.8,5.8,4.8)^\top,\nonumber\\
\|p_i-p_T\|&\geq0.45,qquad
d_R(R_{i,d},R_i)\leq\theta_{\rm chart}=3.0<\pi,
\label{eq:spacecraft_state_constraints}
\end{align}
and the input boxes stated below. These local constraints are checked at all prediction
stages but do not require neighbor-budget tightening.

The three-agent main case uses a complete graph, $H=8$, $h=0.14$, mean motion
$0.012$, $d_{\rm safe}=0.70$, $d_{\rm comm}=2.81$,
$\theta_{\max}=0.40$ rad, and $\ell_R=2.0$. Consequently,
\begin{equation}
L^{\rm col}=L^{\rm com}=1,\qquad L^R=0.5,
\end{equation}
and $\epsilon_{ij}^q=10^{-5}$. Input limits are
$|a_{i,l}|\leq0.36$ and $|\tau_{i,l}|\leq0.95$. The main adaptive run uses
$\rho(0)=0.05$, $\rho_{\min}=0.002$, $\delta_\rho=0.001$, and seed 11. SLSQP uses
40 major iterations and function tolerance $10^{-6}$. With final hard-constraint tolerance
$\tau_{\rm num}=10^{-6}$ and $L_{\max}:=\max_{i,j,q}L_{ij}^q=1$, the tightened and trust checks leave the
pairwise offset $\epsilon_{ij}^q-(1+L_{\max})\tau_{\rm num}=8\times10^{-6}>0$;
shifted reserves are rechecked after rollout.

The initial states are generated reproducibly as
\begin{align}
p_i(0)&=s_0p_{i,d}+b_0+\sigma_p z_{p,i},&
v_i(0)&=\sigma_v z_{v,i},\nonumber\\
R_i(0)&=R_{i,d}\Exp(\delta\theta_i),&
\Omega_i(0)&=\sigma_\Omega z_{\Omega,i},
\label{eq:initial_conditions}
\end{align}
where the $z$ vectors are fixed-seed standard normal samples and
$\|\delta\theta_i\|$ is bounded by the reported initial-angle limit. In the main case,
$s_0=1$, $b_0=(0.05,-0.03,0.02)$, $(\sigma_p,\sigma_v,\sigma_\Omega)=
(0.05,0.02,0.03)$, and the attitude limit is $0.15$ rad. Twelve closed-loop samples
are used. The challenge changes only the parameters stated in Section~\ref{sec:challenge},
uses smaller initialization noise, and is simulated for 16 samples.

The time-invariant stage weights for position, velocity, attitude, and angular velocity
are $(3.0,2.2,6.0,0.5)$; relative position and attitude use $(6.5,0.2)$.
Acceleration, torque, and input slew use $(0.08,0.05,0.12)$, while the own-shift
regularizer uses $(0.8,0.15)$ for position and attitude. Relative-velocity, homotopy,
overshoot, and soft geometric penalties are disabled; the corresponding geometric
conditions are imposed as hard constraints. Configuration files, seeds, per-step logs,
and scripts are included in the source archive.

The terminal matrix $P$ and feedback $K$ are obtained from the local 12-dimensional
position--velocity--log-attitude--angular-velocity linearization and the discrete Riccati
equation. The implemented terminal controller is the equilibrium feedforward plus
$-Ke$, projected onto the hard input box. With $P_i=3.5P_{\rm DARE}$, the terminal
cost is $e^\top P_i e$ and the hard terminal set is $e^\top P_i e\leq0.15$.
No terminal, trust, collision, or
communication slack is permitted.

Let $A_i,B_i$ denote the exact equilibrium Jacobians of the optimizer's sampled error
map; the numerical design approximates them by centered differences with step $10^{-6}$. Thus
$e_i^+=A_ie_i+B_i\delta u_i+O(\|(e_i,\delta u_i)\|^2)$. Let
$Q_i^{\rm D}=Q_i+\Delta Q_i$, $\Delta Q_i\succeq0$, and $R_i^u$ be the terminal-design
weights; the reported $\Delta Q_i$ contains nonnegative relative-state weights. Set
$\ell_i^{\rm D}(e_i,\delta u_i)=e_i^\top Q_i^{\rm D}e_i+
\delta u_i^\top R_i^u\delta u_i$, so $\ell_i^0\leq\ell_i^{\rm D}$. The Riccati solution is
\begin{align}
P_{\rm DARE}={}&Q_i^{\rm D}+A_i^\top P_{\rm DARE}A_i\nonumber\\
&-A_i^\top P_{\rm DARE}B_i
(R_i^u+B_i^\top P_{\rm DARE}B_i)^{-1}
B_i^\top P_{\rm DARE}A_i,
\label{eq:dare}
\end{align}
with
\begin{equation}
K_i=(R_i^u+B_i^\top P_{\rm DARE}B_i)^{-1}B_i^\top P_{\rm DARE}A_i,
\qquad P_i=3.5P_{\rm DARE}.
\label{eq:terminal_lqr}
\end{equation}
Writing $A_{K,i}=A_i-B_iK_i$ and substituting \eqref{eq:terminal_lqr} into
\eqref{eq:dare} yields the closed-loop identity
\begin{equation}
A_{K,i}^\top P_{\rm DARE}A_{K,i}-P_{\rm DARE}
=-(Q_i^{\rm D}+K_i^\top R_i^uK_i).
\label{eq:dare_decrease_identity}
\end{equation}
Consequently the scaled linear terminal pair has the strict surplus
\begin{align}
&V_{f,i}(A_{K,i}e_i)-V_{f,i}(e_i)
+\ell_i^{\rm D}(e_i,-K_ie_i)\nonumber\\
&\qquad=-2.5e_i^\top(Q_i^{\rm D}+K_i^\top R_i^uK_i)e_i\leq0.
\label{eq:linear_terminal_surplus}
\end{align}
An analytic local radius follows from the nonlinear remainder. Let
$\gamma=3.5$, $W_i=Q_i^{\rm D}+K_i^\top R_i^uK_i$, and, on a ball where the unprojected
feedback is input-admissible, write the exact sampled error map as
\begin{equation}
e_i^+=A_{K,i}e_i+\varphi_i(e_i),\qquad
\|\varphi_i(e_i)\|\leq c_{\varphi,i}\|e_i\|^2.
\label{eq:terminal_remainder}
\end{equation}
Using \eqref{eq:dare_decrease_identity} without dropping the remainder gives, with
\begin{equation*}
a_i=2\gamma\|A_{K,i}^\top P_{\rm DARE}\|c_{\varphi,i},
\qquad b_i=\gamma\|P_{\rm DARE}\|c_{\varphi,i}^2,
\end{equation*}
\begin{align}
&V_{f,i}(e_i^+)-V_{f,i}(e_i)+\ell_i^{\rm D}(e_i,-K_ie_i)\nonumber\\
&=-2.5e_i^\top W_ie_i
+2\gamma e_i^\top A_{K,i}^\top P_{\rm DARE}\varphi_i
+\gamma\varphi_i^\top P_{\rm DARE}\varphi_i\nonumber\\
&\leq-\bigl[2.5\lambda_{\min}(W_i)-a_ir-b_ir^2\bigr]\|e_i\|^2
\label{eq:nonlinear_terminal_bound}
\end{align}
for $\|e_i\|\leq r$. A sufficient condition for $\ell_i^{\rm D}$ decrease, and hence
for the $\ell_i^0$ decrease in Assumption~\ref{ass:terminal}, is
\begin{equation}
a_ir+b_ir^2
\leq2.5\lambda_{\min}(W_i).
\label{eq:terminal_radius_condition}
\end{equation}
The implemented law is
$\kappa_i(x_i)=\Pi_{\calU_i}(u_{i,d}-K_ie_i)$, where
$u_{i,d}=\col(-a_{\rm HCW}(p_{i,d},0),0)$ and $\Pi_{\calU_i}$ is componentwise
projection. Equation~\eqref{eq:linear_terminal_surplus} explains the $3.5$ scaling but
does not replace the nonlinear check: Assumption~\ref{ass:terminal} concerns the projected
nonlinear law, and the separate verification below evaluates the exact sampled map after
projection. The complete network terminal inequality, including relative terms, remains
the analytical condition in Assumption~\ref{ass:value_terminal}.
\end{revision}

\subsection{Main comparison}

Four methods use identical models, initial conditions, constraints, costs, horizon, warm
start, and solver tolerance: the proposed rule \eqref{eq:data_budget}; a five-sample
envelope using $\max_{t-5\leq s<t}\mu_i(s)$; a fixed $\rho_i=0.05$; and a
trajectory-only controller without a Lipschitz margin. The last has no residual-certificate
premise. Each solve checks the optimizer, committed shift, and terminal warm start with
zero slacks; among candidates satisfying every hard inequality within $10^{-6}$,
the smallest-cost one is accepted. Solver termination and certification are logged
separately.

\begin{revision}
In the numerical tables, $\bar n_{\rm it}$ denotes the mean solver iteration count and
$\mathcal R_{\min}:=\min_t\mathcal R(t)$. The column ``cert.'' records whether all
reported hard checks are satisfied.
\end{revision}

\begin{table}[t]
\color{revisionblue}
\centering
\caption{Main case ($N=3$, $H=8$). Here $\mathcal R_{\min}$ is the minimum
shifted-fallback reserve, and ``cert.'' requires all hard, packet, terminal, model,
actual-edge, and fallback checks.}
\label{tab:main_comparison}
\scriptsize
\setlength{\tabcolsep}{2.5pt}
\begin{tabular}{@{}lcccccc@{}}
\toprule
Method & $\bar\rho$ & $\bar n_{\rm it}$ & $e_p(T)$ & $e_R(T)$ & $\mathcal R_{\min}$ & cert. \\
\midrule
Adaptive & 0.0090 & 8.78 & 0.0035 & 0.0019 & 0.0709 & yes \\
Windowed & 0.0189 & 8.58 & 0.0035 & 0.0019 & 0.0709 & yes \\
Fixed & 0.0500 & 8.50 & 0.0035 & 0.0019 & 0.0427 & yes \\
Trajectory-only & -- & 8.50 & 0.0035 & 0.0019 & -- & no \\
\bottomrule
\end{tabular}

\end{table}

\begin{figure}[t]
\color{revisionblue}
\centering
\includegraphics[width=0.92\columnwidth]{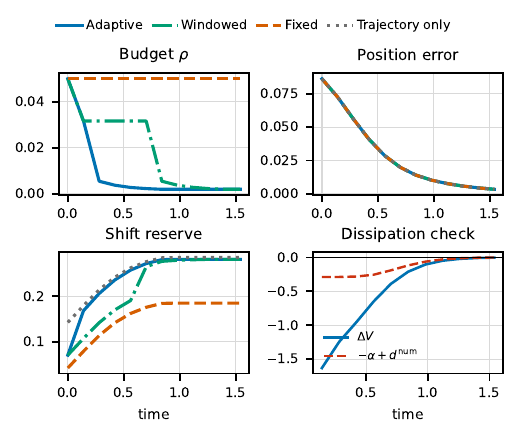}
\caption{Main case. The panels show the announced budget, mean position error,
shifted-fallback reserve, and the adaptive-run dissipation verification
$\Delta V\leq-\alpha+d^{\rm num}$.}
\label{fig:main_comparison}
\end{figure}

\begin{revision}
Table~\ref{tab:main_comparison} and Fig.~\ref{fig:main_comparison} show similar tracking,
as expected in this non-active case. The adaptive mean budget is $0.00901$, $52.3\%$
below the five-sample envelope and $82.0\%$ below the fixed radius. Every certificate
passes, all nominal one-step errors are zero, $\mu_i(t)\leq\rho_i(t)$, and the maximum
terminal-set residual is $-0.1221$. The full-state displacement $\chi_i$, stored in
the logs as \texttt{mu\_cost}, has maximum $0.1049$; it is used only in the value
analysis, not in the pairwise tightening.

The fourth panel verifies Corollary~\ref{cor:terminal_residual}. For
each transition, the shift is reevaluated with new and shifted-old packets, giving
$d^{\rm num}=\sum_i|\Delta_{J,i}|+\zeta$. Over 11 transitions, the maxima of the packet
term, $\zeta$, and $d^{\rm num}$ are $0.04280$, $0.00717$, and $0.04945$. The identity
error is below $3.4\times10^{-16}$ and
$\max(\Delta V+\alpha-d^{\rm num})=-1.79\times10^{-3}$. This is transition-by-transition finite-run
evidence, not a proof. Since $\alpha\geq0.5\|e\|^2$, the associated strict-decrease
threshold is $\sqrt{0.04945/0.5}=0.3145$.
\end{revision}

\subsection{Boundary challenge and adaptive baseline}
\label{sec:challenge}

\begin{revision}
The challenge uses formation radius $1.55$, $d_{\rm comm}=2.85$, and initial scale
$0.95$. The adaptive rule starts at $\rho=0.03$; the fixed rule is an uncapped offline
stress radius $0.08$, while Table~\ref{tab:main_comparison} retains the matched fixed
baseline. Both find hard-feasible optimizers, but the fixed budget exceeds its cap at 14
instants. Hence $\mathcal R_{\min}<0$ means loss of the stored fallback certificate, not
infeasibility of the new optimizer.
\end{revision}

\begin{table}[t]
\color{revisionblue}
\centering
\caption{Boundary case ($N=3$, $H=6$, $d_{\rm comm}=2.85$). The last column tests
the stored shift when all new NLP solutions are unavailable by construction at zero-based
step $t=2$.}
\label{tab:challenge}
\footnotesize
\setlength{\tabcolsep}{3pt}
\begin{tabular}{@{}lccccc@{}}
\toprule
Method & $\bar\rho$ & $\mathcal R_{\min}$ & $\bar J$ & $e_p(T)$ & shift at $t=2$ \\
\midrule
Adaptive & 0.0049 & 0.1219 & 0.4122 & $3.18\times10^{-5}$ & pass \\
Fixed & 0.0800 & -0.0020 & 0.4237 & $2.43\times10^{-3}$ & fail \\
\bottomrule
\end{tabular}

\end{table}

\begin{figure}[t]
\color{revisionblue}
\centering
\includegraphics[width=0.91\columnwidth]{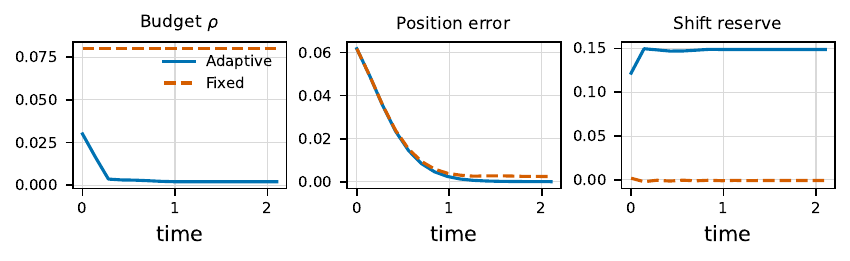}
\caption{Boundary case: announced budget, mean position error, and shifted-fallback
reserve for the adaptive and uncapped fixed rules.}
\label{fig:challenge}
\end{figure}

\begin{revision}
The adaptive run has $\mathcal R_{\min}=0.1219$, $2.7\%$ lower mean value, and final
position error $3.19\times10^{-5}$; the fixed values are $-2.01\times10^{-3}$ and
$2.46\times10^{-3}$. At zero-based step $t=2$, all newly returned NLP candidates are rejected.
The adaptive stored shift passes, completes all 16 certified steps, and finishes at
$2.93\times10^{-5}$; the fixed packet at that solve cannot certify its fallback. This
separates the stored-shift guarantee from feasibility of a new optimizer.
\end{revision}

\subsection{Terminal-set verification, repeated seeds, and network size}
\label{sec:terminal_verification}

\begin{revision}
Independently of the rollout, seed 20260801 generates 20,000 joint samples uniformly in
the product terminal ellipsoid. The projected controller is checked for successor-set
membership, terminal decrease, input/state constraints, and one-layer coupling reserve
at both current and successor states.

For every sampled $x_i\in\calX_{f,i}$, the exact nonlinear map with input projection is
used to evaluate
\begin{align}
r_{f,i}^{+}&=V_{f,i}(f_i(x_i,\kappa_i(x_i)))-c_i,\nonumber\\
r_{V,i}^{\rm D}&=V_{f,i}(f_i(x_i,\kappa_i(x_i)))-V_{f,i}(x_i)
+\ell_i^{\rm D}(e_i(x_i),\kappa_i(x_i)-u_{i,d}).
\label{eq:terminal_verification_residuals}
\end{align}
Here $r_{f,i}^{+}\leq0$ checks invariance, while $r_{V,i}^{\rm D}\leq0$ conservatively
checks the required local decrease because $\ell_i^0\leq\ell_i^{\rm D}$. The
test addresses Assumption~\ref{ass:terminal}; the network value condition in
Assumption~\ref{ass:value_terminal} remains analytical and is not inferred from samples.
Each rollout separately logs packet validity, hard-NLP acceptance, shifted reserve,
actual coupling margins, terminal residuals, and model consistency. The flag
\texttt{all\_steps\_certified} is true only when all applicable checks pass; optimizer
feasibility and preservation of the next shifted candidate remain separate fields.

At the sample level, none of the 20,000 joint states violates a terminal check. The
largest successor-set and terminal-decrease residuals are $-5.37\times10^{-2}$ and
$-7.41\times10^{-3}$, and the minimum communication and relative-attitude reserves are
$0.2465$ and $0.3227$. Fifty-one of 60,000 agent samples activate acceleration saturation,
confirming that the verification evaluates the projected nonlinear law when its input
projection is active rather than implicitly assuming an unsaturated LQR law. At the
rollout level, five further runs use seeds $\{3,7,11,19,23\}$ under the same moderate
initial-error bounds, and all five remain certified. Across these runs, the smallest
communication margin and shifted reserve are $0.2928$ and $0.2032$; the largest
successor-set residual, terminal-decrease residual, and final mean position error are
$2.40\times10^{-8}$, $-5.93\times10^{-5}$, and $0.0206$, respectively. The small
positive successor-set residual is well below the declared $10^{-6}$ numerical
acceptance tolerance.

\begin{figure}[t]
\color{revisionblue}
\centering
\includegraphics[width=0.93\columnwidth]{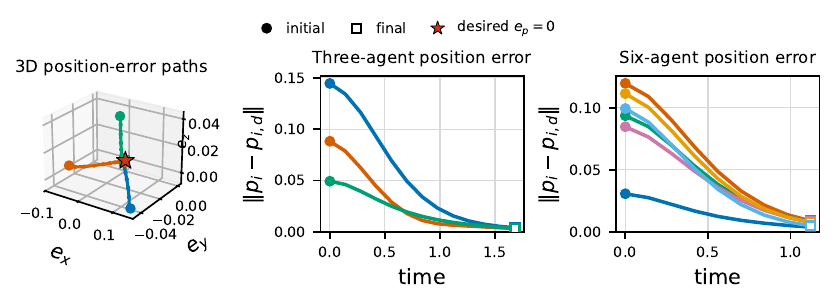}
\caption{Adaptive closed-loop trajectories. Circles, squares, and the red star denote
initial states, final states, and the desired error origin; right panels show every
agent's position-error norm.}
\label{fig:trajectories}
\end{figure}

Network size is evaluated by comparing the three-agent complete graph ($H=8$) with the
six-agent ring ($H=6$). They require, respectively, $5.94$ and $8.11$ seconds per
step, with mean solver iteration counts $8.78$ and $9.23$; both runs are certified, with
minimum reserves $0.0709$ and $0.2576$. These unoptimized Python/SLSQP timings are not
real-time claims: finite-difference NLP solves dominate, while the scalar cap scan is
negligible. Figure~\ref{fig:trajectories} shows the corresponding adaptive closed-loop
trajectories using $e_{p,i}=p_i-p_{i,d}$ and independently scaled axes to make the 3-D
motion visible; dots mark sampled positions. The final mean errors are $0.00351$ for
three agents and $0.00674$ for six.
\end{revision}

\section{Discussion}

\subsection{Relation to neighboring DMPC formulations}

\begin{revision}
Related DMPC methods certify different objects. Giselsson--Rantzer adapt tightening to
a finite dual-decomposition stopping test, whereas K\"ohler--M\"uller--Allg\"ower bound
inexact-dual suboptimality and violation
\cite{GiselssonRantzer2014,KohlerMullerAllgower2019}. Trodden's parallel update uses
tubes for uncertain linear systems, while the mutual-disturbance method exchanges
optimized disturbance and constraint sets \cite{Trodden2014,TroddenMaestre2017}.
Those bounds concern iteration errors or reachable sets within dual or tube updates;
ours concerns accepted coupling-output motion relative to a stored shift.

Accordingly, $\rho_i$ is imposed before each manifold NLP and $\mu_i$ is computed
afterward from two trajectories. No inner dual iteration, convex master problem, or
disturbance-set exchange is used. The matched five-sample envelope changes only the
update rule and gives $\bar\rho=0.0189$ versus $0.0090$ with the same hard certificate;
the fixed and trajectory-only baselines isolate adaptation and the residual margin.
\end{revision}

\subsection{Communication and computation}

Each sample has two pre-solve communication rounds and one post-solve round: shifted
trajectories are exchanged, budgets are announced, and accepted trajectories are then
returned. Relative to trajectory-only DMPC this adds one trajectory and two scalars,
but no inner iteration; $\chi_i$ is recomputed from trajectories already sent. All caps cost
$O(H\sum_{(i,j)\in\calE}n_q(i,j))$ scalar evaluations and add no decision variable;
finite-difference SLSQP dominates the reported prototype timing.

\subsection{Scope and Extensions}

Synchronous exchange makes packet timing explicit and every certificate directly
checkable. Time stamps identify stale packets; a missing exchange invokes a fail-safe
unless a known packet-age error is included in the nonideal
residual~\eqref{eq:nonideal_residual}. This interface offers a direct path to loss-tolerant,
asynchronous, and changing-graph extensions.

\begin{revision}
The present theorem and numerical claims remain synchronous; loss, delay, and
asynchronous operation are not presented as certified cases under these assumptions.
\end{revision}

The metric construction transfers naturally to Euclidean systems, while its $\SO$
instance retains group propagation, a geodesic packet metric, logarithmic terminal
coordinates, and intrinsic relative-attitude coupling. The guarantees are local,
consistently with the topology of $\SO$; when reserve is small, $\rho_i=0$ preserves the
certified shifted fallback.

The analytical terminal conditions are supported by 20,000 joint nonlinear samples,
including active input projection, with no sampled violation. Mission-specific continuum
verification can build on this design through ellipsoidal partitions, interval analysis,
or sums of squares.

\section{Conclusion}

This paper introduced prediction-budget DMPC with pre-solve enforceable margins,
post-solve verifiable updates, and an edge-wise cap that preserves a shifted fallback.
The local guarantees use explicit distance and attitude constants, hard terminal/safety
constraints, and separate nominal and nonideal execution statements. The spacecraft
study shows reduced tightening in the main case and lower error plus certified forced
fallback at the boundary. The framework also supports future asynchronous and richer
certificate extensions.


\end{document}